\documentclass[acmtog]{acmart}

\makeatletter
\def\@ACM@checkaffil{% Only warnings
    \if@ACM@instpresent\else
    \ClassWarningNoLine{\@classname}{No institution present for an affiliation}%
    \fi
    \if@ACM@citypresent\else
    \ClassWarningNoLine{\@classname}{No city present for an affiliation}%
    \fi
    \if@ACM@countrypresent\else
        \ClassWarningNoLine{\@classname}{No country present for an affiliation}%
    \fi
}
\makeatother

\AtBeginDocument{%
  \providecommand\BibTeX{{%
    \normalfont B\kern-0.5em{\scshape i\kern-0.25em b}\kern-0.8em\TeX}}}

\usepackage{amsmath,amsfonts}
\usepackage{graphicx}
\usepackage{textcomp}
\usepackage{multicol}
\usepackage{multirow}
\usepackage{xcolor}
\usepackage[caption=false]{subfig}
\usepackage{color}
\usepackage{booktabs}
\usepackage{tabularx,listings,lstautogobble,color,xspace}
\usepackage{comment}
\usepackage{hyperref}
\usepackage{soul}
\usepackage{enumitem}
\usepackage{algorithm}
\usepackage[noend]{algpseudocode}
\usepackage{bm}
\usepackage{bbold}
\newcommand{\B}[1]{\bm{#1}}
\definecolor{darkBlue}{rgb}{0,0,0.75}
\definecolor{darkred}{rgb}{0,0.5,0}
\newtheorem{theorem}{Theorem}[section]
\newtheorem{proposition}[theorem]{Proposition}
\newcommand{\mina}[1]{#1}
\newcommand{\caleb}[1]{\textcolor{blue}{#1}}
\newcommand{\serif}[1]{\textcolor{red}{#1}}

 \usepackage[normalem]{ulem} % for sout

\newcommand{\matching}{\textit{match}\xspace}
\newcommand{\bimatching}{\textit{bmatch}\xspace}
\newcommand{\vcover}{\textit{vcover}\xspace}
\newcommand{\domset}{\textit{dom-set}\xspace}
\newcommand{\densest}{\textit{dense-sub}\xspace}
\newcommand{\opttrans}{\textit{opt-trans}\xspace}
\newcommand{\cplex}{\textit{CPLEX}\xspace}
\newcommand{\gurobi}{\textit{Gurobi}\xspace}
\newcommand{\mwuomp}{\textit{MWU-opt}\xspace}
\newcommand{\mwupetsc}{\textit{MWU-PETSc}\xspace}
\newcommand{\mwuhybrid}{\textit{MWU-opt}\xspace}
\newcommand{\graft}{\textit{ms-bfs-graft}\xspace}
\newcommand{\gbbs}{\textit{GBBS}\xspace}
\newcommand{\mtxname}[1]{\textit{#1}\xspace}

\setcopyright{none}
\begin{document}

\title{Efficient Parallel Implementation of the Multiplicative Weight Update Method for Graph-based Linear Programs}

\author{Caleb Ju}
\email{cju33@illinois.edu}
\affiliation{%
  \institution{Georgia Institute of Technology}
  % \streetaddress{P.O. Box 1212}
  % \country{USA}
}

\author{Serif Yesil}
\affiliation{%
  \institution{University of Illinois at Urbana-Champaign}
  % \streetaddress{P.O. Box 1212}
  \city{Atlanta}
  \state{GA}
  % \country{USA}
}

\author{Mengyuan Sun}
\email{ms65@illinois.edu}
\affiliation{%
  \institution{University of Illinois at Urbana-Champaign}
  % \streetaddress{1 Th{\o}rv{\"a}ld Circle}
  % \city{Hekla}
  % \country{Iceland}
}

\author{Chandra Chekuri}
\email{chekuri@illinois.edu}
\affiliation{%
  \institution{University of Illinois at Urbana-Champaign}
  % \city{Rocquencourt}
  % \country{France}
}

\author{Edgar Solomonik}
\email{solomon2@illinois.edu}
\affiliation{%
  \institution{University of Illinois at Urbana-Champaign}
  % \country{India}
}

\renewcommand{\shortauthors}{Ju et. al.}

\renewcommand{\sout}[1]{}
\begin{abstract}
    Positive linear programs (LPs) model many graph and operations research problems. One can solve for a $(1${+}$\epsilon)$-approximation for positive LPs, for any selected $\epsilon$, in polylogarithmic depth and near-linear work via variations of the multiplicative weight update (MWU) method. Despite extensive theoretical work on these algorithms through the decades, their empirical performance is not well understood.

    In this work, we implement and test an efficient parallel algorithm for solving positive LP relaxations, and apply it to graph problems such as densest subgraph, bipartite matching, vertex cover and dominating set. We accelerate the algorithm via a new step size search heuristic. Our implementation uses sparse linear algebra optimization techniques such as fusion of vector operations and use of sparse format. Furthermore, we devise an implicit representation for graph incidence constraints. We demonstrate the parallel scalability with the use of threading OpenMP and MPI on the Stampede2 supercomputer. We compare this implementation with exact libraries and specialized libraries for the above problems in order to evaluate MWU's practical standing for both accuracy and performance among other methods. Our results show this implementation is faster than general purpose LP solvers (IBM CPLEX, Gurobi) in all of our experiments, and in some instances, outperforms state-of-the-art specialized parallel graph algorithms.
\end{abstract}

\maketitle

\section{Introduction}
\sout{Large-scale graph processing is invaluable in many applications in areas like network science, machine learning, and neuroscience ~\cite{wang2014,brin1998,kleinberg1999,bassett2017network,pfaff2020learning,alguliyev2021graph}.}
%\edgar{consider using more generic qualification of area than 'pandemic modeling'} \caleb{added a new description}
%in sGraph algorithms are useful in a variety of application domains, such as processing of social networks
%d to analyze social, web, biological, and e-commerce networks with many useful applications in machine learning, network sciences~\cite{wang2014,brin1998,kleinberg1999}.
%However,
% \chandra{It would be good to add more areas/citations I believe.}
Designing general graph libraries for massively-parallel machines is challenging as combinatorial graph algorithms often have limited concurrency and arithmetic intensity~\cite{lumsdaine2007challenges}.
%It is challenging to design a general purpose fast parallel framework for graph applications due to the irregular structure of graphs, algorithms' low arithmetic intensity and lack of concurrency~\cite{lumsdaine2007challenges}.
%Although, well-designed primitives ~\cite{blelloch1989scans,galois,bulucc2011combinatorial,shun2013ligra, solomonik2017scaling, bulucc2009parallel} and compressed data structures~\cite{shun2015smaller,dhulipala2019low} express many such algorithms efficiently, improving parallel scalability, especially in the distributed setting, continues to be an active area of research.
Many have parallel depth at least proportional to graph diameter, though in practice these algorithms contain sufficient concurrency to achieve high-efficiency on shared-memory machines~\cite{galois,shun2013ligra,shun2015smaller,dhulipala2019low}.
Designing efficient distributed-memory parallelizations is more challenging, though it has been achieved with use of graph partitioning~\cite{sui2010parallel}, formulations via sparse matrix products~\cite{bulucc2011combinatorial,bulucc2009parallel,solomonik2017scaling}.

\mina{Another method for solving graph problems is to reformulate or relax \cite{williamson2011design,vazirani2013approximation} the graph problem into a linear program (LP).}
%To solve an LP, one can use any general purpose LP solver, which typically uses simplex or interior point methods. hile these algorithms are efficient in the sequential setting, they have limited parallelism ~\cite{cplex2009v12, gurobi}. 
% \sout{Linear programs (LPs) enable exact and approximate solutions of many important graph problems.}
Recent theoretical breakthroughs
in both sequential and parallel algorithms for several graph problems~\cite{andoni2020parallel,li2020faster, van2020bipartite,van2020solving,sherman2013nearly,liu2020faster}
indicate that carefully designed LP solvers can enable algorithms with tunable accuracy that have lower cost and depth than their combinatorial counterparts.
Given a matrix $\B{A} \in \mathbf{R}^{m \times n}$, vector $\B{c} \in \mathbf{R}^n$, and vector $\B b \in \mathbf{R}^m$, an LP 
is the optimization problem,
\begin{equation} \label{eq:LPDef} 
    \max_{\B x \in \mathbf{R}^n} \langle \B c, \B x \rangle \text{ s.t. }
    \B{Ax} \leq \B b,
\end{equation}
where $\langle \B c, \B x \rangle = c_1 x_1 + c_2x_2 + \ldots  + c_nx_n$,
and $\B u \leq \B v$ means $u_i \leq v_i \ \forall i$.
\sout{LP relaxations also are a standard and powerful tool in the development of approximation algorithms and heuristics
for problems that do not have efficient exact algorithms} \cite{williamson2011design,vazirani2013approximation}.

Typically, to solve a general LP, one employs variants of the simplex or
interior point methods.
While these algorithms are efficient in the sequential setting, they often have limited parallelism.
LP solving is P-Complete, which implies
that a poly-logarithmic depth parallel algorithm is believed unlikely to
exist~\cite{karp1989survey}.

However, many graph problems can be solved via LPs that contain only positive entries, and
these class of LPs admit efficient parallelizable solvers if some approximation is allowed.
\sout{ We consider approximation
algorithms to the standard feasibility mixed packing and covering LP~\cite{young2001sequential, mahoney2016approximating},
\begin{equation} \label{eq:NormFeasibleLP} 
    \exists \ {\B x \in \mathbf{R}^n} 
    \text{ s.t. } \B{Px} \leq \mathbb{1}, \B{Cx} \geq \mathbb{1}, \B x \geq \mathbb{0},
\end{equation}
where $\B P$ and $\B C$ are \emph{nonnegative}. Vectors $\mathbb{1}$ and $\mathbb{0}$ are the all ones and zeros vector, respectively.
The algorithms we consider seek a $(1+\epsilon)$-relative approximation: that is, a solution $\B x$ such 
that $\B{Px} \leq (1+\epsilon)\mathbb 1$ and $\B{Cx} \geq \mathbb{1}$.
% \chandra{Do we want to talk about pure packing and covering here? I added  the following sentence for now.}
Two important special cases are pure packing and pure covering LPs that we describe
formally later.}
\sout{There is extensive work on approximating positive LPs in the sequential and parallel setting; we
refer the reader to \cite{arora2012multiplicative,quanrud2019thesis,wang2019thesis}.
The first efficient parallel approximation algorithms for positive LPs was developed by Luby and Nisan~\cite{luby1993parallel} which was refined later by work of Young \cite{young2001sequential}.}
These approaches yield a parallel algorithm with depth that is a polynomial in $1/\epsilon$ and $\log n$ (the number of variables), and is independent of the structure of the constraint matrices $\B P$ and $\B C$. \mina{ It outputs a result that is a $(1+\epsilon)$-approximation, which, for a maximization problem entails a relative error of $(1-\epsilon)$ in the objective value, and for a minimization problem, a relative error of $(1+\epsilon)$.}

The polynomial scaling with $1/\epsilon$ has been improved in recent work. Our implementation is based on an accelerated version of a more recent algorithm with an improved depth of $\tilde{O}(\epsilon^{-3})$ for mixed problems and $\tilde{O}(\epsilon^{-2})$ for pure covering or pure packing problems. \mina{The only other previously published study ~\cite{makari2013distributed} of a distributed and parallel LP solver built using this approach implemented and compared run times of two earlier algorithms: MPCSolver and Young's algorithm, with $\tilde{O}(\epsilon^{-5})$ and $\tilde{O}(\epsilon^{-4})$ depth, respectively.} Further details on theoretical developments in approximately solving positive LPs in parallel are provided in Section~\ref{sec:MWUDef}.

The algorithm we focus on is from Mahoney et. al.~\cite{mahoney2016approximating}, which offers the fastest theoretical performance for pure packing, pure covering, and mixed packing and covering LPs in a parallel setting. We describe these problems and their associated graph problems formally in Section \ref{sec:MWUProbs}. This algorithm we call MWU since it utilizes the weights vector from the more general MWU method~\cite{arora2012multiplicative}. Despite the algorithm's strong theoretical foundation, little is known about its empirical performance, especially for solving real-world graph problems over large datasets. Furthermore, understanding how this algorithm compares to other LP solvers and specialized graph libraries is an open question.

\mina{In this paper, we create a practical and scalable implementation of a parallel (1+$\epsilon$)-approximate positive LP solver by adapting the MWU algorithm with the current fastest theoretical performance, and we also provide a comprehensive empirical study comparing MWU against exact solvers, specialized libraries, and previous related work. We also test a line search method that finds the largest step-size permitted without violating theoretical guarantees.
This step-size search reduces in the number of overall MWU iterations in practice by multiple orders of magnitude.
Our results suggest that MWU can be fast in both theory and practice, with run times comparable to even specialized libraries but with the added benefit of, one, being generalizable and, two, more versatile, as it can benefit from high-performance sparse linear algebra libraries. We believe this marks MWU as a viable alternative for solving graph problems approximately on a large dataset.} \sout{In this paper, we implement a practical, parallel, approximate positive LP solver, and we conduct extensive numerical experiments on a supercomputer. Our results suggest that MWU can be fast in both theory and practice.
} Overall, our paper makes the following contributions:

\begin{itemize}[leftmargin=*]
    \item the first shared and distributed-memory implementation of a general-purpose approximate solver based on Mahoney et. al.~\cite{mahoney2016approximating} for positive LPs, called MWU
    % \textcolor{red}{edgar: This does not make novelty clear, is ours the first parallel threaded or distributed implementation? If not, this is not a contribution.}
    \item a step size search method for MWU that empirically reduces the iteration count by up to three orders of magnitude without adding significant overhead
    \item an efficient implementation of MWU using an implicit representation of common constraint matrices observed in graph LPs as well as other standard techniques for SLA. We show these optimizations provide good scalability and lead up to 5.2x speedup relative to MWU implemented with PETSc (a library for parallel sparse linear algebra).
    \item the first comparison of an approximate positive LP solver against general LP solvers as well as specialized parallel graph algorithms. In particular, MWU finds $(1 + \epsilon)$-relative ($\epsilon =0.1$) solutions up to 3-2800x and 5-1180x faster than CPLEX and Gurobi (which find exact solutions for implicit ILPs and exact fractional solutions for relaxed LPs), respectively, for solving several graph problems on large real-world graphs on the Stampede2 supercomputer using KNL compute nodes.
\end{itemize}

\section{Background on Linear Program Solvers}
\label{sec:MWUDef}
% We introduce both the positive LP and MWU algorithm, and show how to use the latter to solve the latter. 
\subsection{Positive LP Solvers} \label{sec:pos_lp_solvers}
% \chandra{I edited this subsection.}
Fast approximate solvers for positive LPs in the sequential settings have been developed since the early 1990's \cite{plotkin1995fast,GrigoriadisK94}, and there is extensive and continued attention to this line of work. We mostly focus on parallel algorithms and refer the reader to \cite{quanrud2019thesis,wang2019thesis} for extensive pointers.
Luby and Nisan provide the first 
parallel algorithm for explicit positive LPs which obtain a $(1+\epsilon)$-approximation in  $\tilde{O}(\epsilon^{-4})$ 
iterations\footnote{We write $\tilde{O}(f(n))$ to be proportion to $f(n)$ and a polylogarithmic of $f(n)$, i.e., $\tilde{O}(f(n)) \propto O(f(n)\log^{O(1)}(f(n)))$} 
for pure
packing and covering LPs~\cite{luby1993parallel}. Young clarified and extended
this work to solve the more general mixed packing and covering LPs in
parallel in $\tilde{O}(\epsilon^{-4})$ iterations~\cite{young2001sequential}. 
The dependence on $\epsilon$ has remained unchanged for over 10 years until the work of 
Allen-Zhu and 
Orecchia, who 
% use techniques from optimization of smooth functions to
solve pure packing and pure covering LPs in parallel with 
$\tilde{O}(\epsilon^{-3})$ iterations~\cite{allenzhu2016using}.
They also obtained better dependence on $\epsilon$ in the sequential setting~\cite{allen2019nearly} for pure packing and covering LPs (see also~\cite{mahoney2016approximating}).
%as well as pure packing and pure covering LPs sequentially in 
%$\tilde{O}(\epsilon^{-1})$ and  $\tilde{O}(\epsilon^{-1.5})$ iterations, %respectively~\cite{allen2019nearly}. 
%Wang et. al.~\cite{mahoney2016approximating} reduce the latter iteration count to %$\tilde{O}(\epsilon^{-1})$~\cite{wang2015unified}. 

Recently, Mahoney et. al.~\cite{mahoney2016approximating} utilized ideas of Young~\cite{young2014nearly} on faster near-linear time sequential solver to develop new parallel algorithms for positive LPs. 
For mixed packing and covering LPs, their algorithm converges
in $ O\big(\log(m_p+m_c)\log(n/\epsilon)/\epsilon^3\big)$ iterations and is also work-efficient, i.e., 
work is near linear to the number of nonzeros in the LP. For pure packing 
and pure covering LPs, the dependence on $\epsilon$ is $\epsilon^{-2}$ instead.
When $\epsilon$ is not too small, these algorithms have low depth in the PRAM model.

Empirical studies of these fast approximate algorithms have been mainly limited to relatively small problems.
%is much more sparse. 
Koufogiannakis and Young adapt a sequential mixed packing and covering LP
and show it outperforms simplex on randomly generated binary matrices with 
dimension up to 2222~\cite{koufogiannakis2014nearly}. Allen-Zhu and Orecchia 
compare their $\tilde{O}(\epsilon^{-1})$-dependent sequential algorithm~\cite{allen2014using} to the algorithms of Luby and Nisan~\cite{luby1993parallel} and Awerbuch and Khandekar~\cite{awerbuch2009stateless}
for solving pure packing LPs with a randomly generated matrix of size $60 \times 40$. Jelic et.\ al.\ implement a parallel primal-dual method on GPUs to solve positive LPs, although their constraint matrices are randomly generated binary matrices with dimensions up to 25000~\cite{jelic2015fast}. The most closely related work to ours is that of Makari et.\ al.\ \cite{makari2013distributed}, who implement a gradient descent algorithm to solve generalized matching on large real-world and synthetic graphs.
Their implementation, based on the algorithm from \cite{awerbuch2009stateless}, has a $\tilde{O}(1/\epsilon^5)$ number of iterations, and they confine their attention to a single graph problem. 

In this paper we focus on only obtaining fractional solutions to the LP
problems. Since all LPs are polynomial time solvable, whereas
the discrete formulation of some of the graph problems we consider are NP-hard,
this allows us to have a more uniform and fair comparison to prior art
(e.g.,~\cite{makari2013distributed}, as well as state-of-the-art software for
solving LPs like Gurobi).

There also exists rounding techniques (which convert a fractional solution to an
integral one) or specialized algorithms to solve or approximately solve the
discrete problems, but these methods are specific to the problem and have
different levels of parallelism, efficiency, and approximation guarantees. For
example, an exact parallel rounding technique for a maximum matching problem is
implemented in~\cite{makari2013distributed}, but the run time can be three
orders of magnitude longer than solving the LP problem. On the other hand,
rounding techniques for dominating set are compared in~\cite{li2020performance},
but these only return approximately optimal solutions and the rounding is not
parallelized. There are also specialized libraries for the aforementioned graph
problems~\cite{azad2016computing,dhulipala2020graph} (more details can be found
in Section~\ref{sec:exp_setup}). But again, these implementations are tailored
to the problem at hand, hence it would not be fair to compare the run time
against our general-purpose solver in terms of obtaining an integral solution.

\subsection{The MWU Algorithm} \label{sec:mwu_alg_sec}
We now introduce MWU, the algorithm of Mahoney et. al.~\cite{mahoney2016approximating} 
for approximately solving the standard mixed packing and covering LP in parallel. This section does not contain our modifications for improving the empirical performance, which are found in Section ~\ref{sec:AlgImprove}.

First, we start describing how MWU solves the mixed packing and covering feasibility LP ~\cite{young2001sequential, mahoney2016approximating}. We will discuss how to modify this into an optimization problem later in the section. The feasibility LP is:
\begin{equation} \label{eq:NormFeasibleLP} 
    \exists \ {\B x \in \mathbf{R}^n} 
    \text{ s.t. } \B{Px} \leq \mathbb{1}, \B{Cx} \geq \mathbb{1}, \B x \geq \mathbb{0},
\end{equation}
where $\B P$ and $\B C$ are \emph{nonnegative}. Vectors $\mathbb{1}$ and $\mathbb{0}$ are the all ones and zeros vector, respectively.
The algorithms we consider seek a $(1+\epsilon)$-relative approximation: that is, a solution $\B x$ such 
that $\B{Px} \leq (1+\epsilon)\mathbb 1$ and $\B{Cx} \geq \mathbb{1}$.

MWU (Algorithm~\ref{alg:exp1}) ensures that both
packing and convering constraints, $\B{Px} \leq \mathbb{1}$ and $\B{Cx} \geq \mathbb{1}$, are approximately satisfied
%. Using is a natural choice, but they are  
%difficult to mathematically analyze. %Instead, we utilize
by approximating $\max(\B{Px})$ and $\min(\B{Cx})$ with smoothed maximum and minimum functions,
\begin{align*}
  \mathrm{smax}_\eta(\B x) &= \frac{1}{\eta} 
 \log\Big( 
    \sum\limits_{i=1}^n
    \mathrm{exp}(\eta \cdot x_i)
 \Big), \\
 \mathrm{smin}_\eta(\B x) &= -\frac{1}{\eta} 
 \log\Big( 
    \sum\limits_{i=1}^n
    \mathrm{exp}(-\eta \cdot x_i)
 \Big),
\end{align*}
where $\eta>2$ is a smoothing parameter. The MWU algorithm and step size search make use of their gradients,
%the smooth max (smax) and its gradient, which are, respectively,
\begin{align*}
 \nabla \mathrm{smax}_\eta(\B x) &= 
    \frac{
    \mathrm{exp}(\eta \cdot \B x)
    }{
    \langle \mathbb{1}, \mathrm{exp}(\eta \cdot \B x) \rangle 
    }, % , \label{eq:grad1}, 
 \nabla \mathrm{smin}_\eta(\B x) = 
    \frac{
    \mathrm{exp}(-\eta \cdot \B x)
    }{
    \langle \mathbb{1}, \mathrm{exp}(-\eta \cdot \B x) \rangle 
    }.%  \label{eq:grad2}
\end{align*}
%\begin{equation} \label{eq:SmaxDef} \footnotesize
%\begin{split}
% \mathrm{smax}_\eta(\B x) &= \frac{1}{\eta} 
% \log\Big( 
%    \sum\limits_{i=1}^n
%    \mathrm{exp}(\eta \cdot x_i)
% \Big) \\
% \nabla \mathrm{smax}_\eta(\B x) &= 
%    \frac{
%    \mathrm{exp}(\eta \cdot \B x)
%    }{
%    \langle \mathbb{1}, \mathrm{exp}(\eta \cdot \B x) \rangle 
%    }.
%\end{split}
%\end{equation}
%\begin{equation} \label{eq:SminDef} \footnotesize
%\begin{split}
% \mathrm{smin}_\eta(\B x) &= -\frac{1}{\eta} 
% \log\Big( 
%    \sum\limits_{i=1}^n
%    \mathrm{exp}(-\eta \cdot x_i)
% \Big) \\
% \nabla \mathrm{smin}_\eta(\B x) &= 
%    \frac{
%    \mathrm{exp}(-\eta \cdot \B x)
%    }{
%    \langle \mathbb{1}, \mathrm{exp}(-\eta \cdot \B x) \rangle 
%    }.
%\end{split}
%\end{equation}
% The smoothing parameter $\eta$ is selected to satisfy $\eta > 2$. 
%\noindent The smooth min (smin) and its gradient are, respectively,
For more details on these functions, see 
Chapter 2 of~\cite{quanrud2019thesis}.

\begin{algorithm}[h]
\caption{Multi-Update MWU Method for Mixed Packing and Covering LPs}
%Approximately solves standard mixed packing-covering feasibility problem in parallel.
\footnotesize
\begin{algorithmic}[1]
\Procedure{MWU}{$\B{P} \in \mathbf{R}_+^{m_P \times n}$, $\B{C} \in \mathbf{R}_+^{m_C \times n}$, $\epsilon$ }
  % \State{$\text{num\_iter} \gets 0$}
  % \State{$\text{max\_iter} \gets \log(m_p+m_c)  \cdot \log(n/\epsilon)/\epsilon^3$} \label{line:maxiter_def}
  % \State $\text{num\_iter} \gets 0$
  \State{$\eta \gets 10 \log(m)/\epsilon$ where $m: = m_P + m_C$}
  \State{$x_i \gets \frac{\epsilon}{n \| \B{P}_{:,i} \|_\infty } \ \forall i \in [n]$} \label{line:init_vec} \Comment{$\| \B x\|_\infty = \max_i |x_i|$}
  % \State $\B y \gets \B{Px}$
  % \State $\B z \gets \B{Cx}$ 
  % \State If $M$ not given, set $M \gets \mathrm{GuessM}(\B{P},x_0,\epsilon)$
  \While{constraints not approximately satisfied and $\B C \ne \varnothing$} \label{line:alg1_stop_criteria}
    \State {$\B g \gets \B{P}^\mathsf T \nabla \mathrm{smax}_\eta(\B{Px})$} \label{line:smax_alg1}
    \State {$\B h \gets \B{C}^\mathsf T \nabla \mathrm{smin}_\eta(\B{Cx})$} \label{line:smin_alg1}
    \State $d_i \gets \frac{1}{2\eta}\mathrm{max} \{0,1-\frac{g_i}{h_i}\} \cdot x_i \ \forall i$ \label{line:direction_alg1}
    % \State $\B{d}^\star \gets \max(\B d)$$
    \If{$\max(\B{d}) = 0$} \label{line:get_maxd_alg1}
        \State{Return ``INFEASIBLE''} \label{line:ret_inf}
    \EndIf
    % \State $\B d\_\text{sum} \gets \mathrm{sum}(d)$ 
    % \State $\B{dy} \gets \B{Pd}$ 
    % \State $\B{dz} \gets \B{Cd}$
    % \State $\alpha \gets \operatorname{StepSize}(\B d, \B y, \B z, \B{dy}, \B{dz}, \ldots)$ \label{line:find_step_size}
    % \If{$\alpha < 1$}
    %     \State Return ``INFEASIBLE''
    % \EndIf
    \State{$\B x \gets \B x + \B d$ } \label{line:vec_add}
    % \State $\B y \gets \B y + \alpha \cdot \B{dy}$ 
    % \State $\B z \gets \B z + \alpha \cdot \B{dz}$
    \State{$\B C \gets \{ c_i \ : \ c_i^Tx < 1 \}$} \label{line:drop_cons} \Comment{Keep unsatisfied constraints}
    % \State{$\text{num\_iter} \gets \text{num\_iter} + 1$}
  \EndWhile
  % \State $x \gets \mathrm{LastStep}(\B{P},x)$ \Comment If $\B{P}x < \mathbb{1}$, takes one last step
  \State{\Return $x$}
\EndProcedure
\end{algorithmic}
\label{alg:exp1}
\end{algorithm}
% \caleb{CALEB: Added a couple sentences here. I will try to add more tomorrow, but if you have suggestions about what to add, please let me know.}\newline
% \caleb{First, $\eta$ is chosen so that $\mathrm{smax}_\eta$ and $\mathrm{smin}_\eta$ are within an additive $\epsilon$ approximation of $\max$ and $\min$, respectively.} Vector $\B x$ is initialized so that $\B{Px} \leq \epsilon \cdot \mathbb{1}$ (Line~\ref{line:init_vec}). Vectors $\B g$ and $\B h$ are the gradients of the smax and smin of the packing and covering constraints, respectively, and they detail how close we are to violating or satisfying a constraint (Line~\ref{line:smax_alg1},~\ref{line:smin_alg1}). The gradient vectors determine the   \emph{step} or \emph{update direction} vector $\B d$ for updating $\B x + \B d$ (Line~\ref{line:vec_add}), with the goal of satisfying more covering constraints while not violating a packing constraint. \caleb{Unlike previous algorithms~\cite{quanrud2018approximating}, we multiplicatively increase (i.e., $x_i = (1+\delta_i)x_i$ for some $\delta_i \geq 0$) $\B x$ rather than additively (i.e., $x_i = x_i + \delta_i$).} The step direction vector $\B d$ is an approximate solution to the Lagrangian relaxation of~\eqref{eq:NormFeasibleLP}, 
% \caleb{Rewrote this section} 

The algorithm initializes the vector $\B x$ with small values so that each starting packing constraint is at most $\epsilon$ (Line~\ref{line:init_vec}). The smoothing parameter $\eta$ is set so that both $\mathrm{smax}_\eta$ and $\mathrm{smin}_\eta$ are within an $\epsilon$ additive error of max and min, respectively. In each MWU iteration, the algorithm multiplicatively updates $\B x$. This is done by defining a step or update vector, $\B d$, where $d_i$ is a multiple of $x_i$ (Line~\ref{line:direction_alg1}), and adding $\B d$ to $\B x$ (Line~\ref{line:vec_add}). Vectors $\B g$ and $\B h$, which are gradients of the smoothed max packing and min covering constraints, respectively, are also utilized to define $\B d$ (Lines~\ref{line:smax_alg1},~\ref{line:smin_alg1}). In particular, $\B d$ is an approximate solution to the Lagrangian relaxation,
\begin{equation} \label{eq:LagrangianRelax} % \footnotesize
\begin{split}
    \exists \B d \in \mathbf{R}^n_{\geq 0}
    \text{ s.t. } &\langle \B w_p, \B{P d} \rangle = \langle \B{P}^T\B w_p, \B d \rangle \leq 1, \\
    & \langle \B w_c, \B{C d} \rangle  = \langle \B{C}^T\B w_c, \B d \rangle \geq 1,
\end{split}
\end{equation}
where $\B w_p = \nabla \mathrm{smax}_\eta(\B Px)$ and $\B w_c =  \nabla \mathrm{smin}_\eta(\B C x)$.
%the weight vectors $\B w_p$ and 
%$\B w_c$ are defined by the gradients of smax and smin,
%\[ 
%    \B w_p = \frac{
%    \mathrm{exp}(\eta \cdot \B{Px})
%    }{
%    \langle \mathbb{1}, \mathrm{exp}(\eta \cdot \B{Px}) \rangle 
%    } \text{ and }
%    \B w_c = \frac{
%    \mathrm{exp}(-\eta \cdot \B{Cx})
%    }{
%    \langle \mathbb{1}, \mathrm{exp}(-\eta \cdot \B{Cx}) \rangle 
%    }.
%\]

% If the positive LP is infeasible, then the algorithm is guaranteed to set $\B d = \mathbb 0$ at some point~\cite{mahoney2016approximating}, in which case it reports that the solution is infeasible.
If the positive LP is infeasible, then there exists some MWU iterations where $\B d=\mathbb 0$, in which case the algorithm reports the LP is infeasible (Line~\ref{line:get_maxd_alg1})~\cite{mahoney2016approximating}. Assuming otherwise, the theoretical analysis guarantees MWU will return an $(1+\epsilon)$-relative solution. Throughout the algorithm, we drop satisfied covering constraints, as these can unnecessarily slow down progress (Line~\ref{line:drop_cons}). % Consequently, we terminate MWU once all the covering constraints are satisfied or when there exists no covering constraints.
%``INFEASIBLE''
%(Line~\ref{line:ret_inf}).
% MWU drops satisfied covering constraints at the end of each iteration , as these can otherwise unnecessarily slow down progress.
%to ensure constraints are unbounded, i.e., $[\B{Cx}]_i \gg 1$, 
%which can stunt the algorithm's convergence. 
Finally, the algorithm returns $\B x$ when all the covering constraints are satisfied or when there exists no covering constraints.

We note that Algorithm~\ref{alg:exp1} can also solve pure packing or pure covering LPs, which are, respectively,
\begin{align*}
    &\max \ \langle \mathbb{1}, \B x \rangle \text{ s.t. } \B{Px} \leq \mathbb{1}, \B{x} \geq \mathbb{0}, \B{x} \in \mathbf{R}^n \\
    &\min \ \langle \mathbb{1}, \B x \rangle \text{ s.t. } \B{Cx} \geq \mathbb{1}, \B{x} \geq \mathbb{0}, \B{x} \in \mathbf{R}^n.
\end{align*}
For example, to solve a pure packing LP, we embed the objective function as the added constraint, $\frac{1}{M} \mathbb{1}^T\B x \geq 1$, where $M$ is the estimate of the maximum value, e.g., $M=\sum\limits_{i=1}^n \max_{j : p_{ji} > 0} 1/p_{ji}$. Then we do binary search over $M$, using Algorithm~\ref{alg:exp1} to determine if the resulting mixed packing and covering LP is a feasible. Since there is one covering constraint, then $\mathrm{smin}_\eta(\B{Cx}) = \min(\B{Cx})$. This exact approximation permits one to scale the step direction (Line~\ref{line:direction_alg1}) by a factor of $2$ in the theoretical analysis, which improves the number of iterations by a factor of $\epsilon$~\cite{mahoney2016approximating}. Also, noting $\B{C} = \frac{1}{M} \mathbb{1}^T$ and $\nabla \mathrm{smin}_\eta(\B{Cx})=1$, we have $\B{h}=\frac{1}{M} \mathbb{1}$ (Line~\ref{line:smin_alg1}), so we do not need to explicitly compute $\B h$. Solving a pure covering LP is done via a similar transformation. Solving a mixed covering and packing optimization problem, likewise, involves embedding the constraint that corresponds to the direction of optimization.

\section{Graph Problems as Positive LPs} \label{sec:MWUProbs} 

We now consider several graph problems.  We first define integer programming
(IP) formulations which exactly model the underlying graph problem.  We then
obtain an LP by relaxing the integrality constraints. For some problems, such as
bipartite matching and densest subgraph, the solution to the LP relaxation
matches the IP's solution (i.e., the solution is integral), whereas for NP-hard
problems dominating set and vertex cover there is an \textit{integrality gap}
between the solution to the LP relaxation and IP.
See~\cite{vazirani2013approximation,williamson2011design} for the role of LP
relaxations in the development of approximation algorithms, and
also~\cite{li2020performance} for a performance study on rounding an LP
relaxation to an integral solution for dominating set. Our goal
is to design a general-purpose solver, and the design and performance of
rounding schemes are problem dependent (see the end of
subsection~\ref{sec:pos_lp_solvers} for further details). Therefore, we do not
consider rounding in our implementation nor performance comparisons.

% \st{Assume that we are given an undirected graph $G=(V,E)$} 
Let $G=(V,E)$ be an unweighted, undirected graph where $V$ is the set of vertices and $E$ is the set of edges, with $n=|V|$ and $m=|E|$. For simplicity, we assume $G$ has no self-loops.
Note that our formulations can be extended towards weighted graphs as well.
% We assume that our graphs are unweighted and undirected, however, our formulations can be extended to weighted graphs. 
% , or edges going from a vertex to itself.
% \caleb{Should we assume
% unweighted graphs, but mention weighted is easy to do? sure} 
For a vertex $v \in V$, let
$N(v)$ be the neighbor vertices of $v$ ($v$ is not included in $N(v)$), and $\mathrm{inc}(v)$ be the set of edges incident to $v$. 

The neighbor relations between the vertices of $G$ can be represented as an adjacency matrix, $\B{A} \in \{0,1\}^{n \times n}$, which is symmetric and has a nonzero for each edge $e \in E$. 
% Since $G$ is unweighted and undirected, $\B A$ is symmetric and has only 1s as nonzeros. 
The incidence relation between the vertices and the edges can be represented as a vertex-edge incidence matrix, $\B M$, where 
% \st{Let $\B A$ be the adjacency matrix for the graph $G$ of an unweighted and undirected graph. $\B{A}$ is symmetric and all its non-zeros are 1.} 
% We define the vertex-edge incidence matrix $\B M$ to be a $|V| \times |E|$ matrix where,
\begin{equation} \label{eq:incidence} 
    \B{M}_{u,e} = \begin{cases} 
    1 \ : \ u \in e, e \in E\\
    0 \ : \ \text{otherwise}
    \end{cases}, \ \B{M} \in \{0,1\}^{n \times m}.
\end{equation}
% In the \textbf{maximum matching} problem, the objective is to select
% the largest subset $F \subset E$ of edges such that no vertex is 
% incident to more than one edge in the set $F$. 
% \serif{can add more details for the examples.} 
% \begin{verbatim}
% App. High level descr.
% Example in Fig1. LP formulation. Explanation 
% of LP formulation. How does it look like in mtx form?
% \end{verbatim}
% \serif{needs an intro sentence to the apps}
In this work, we consider four graph applications: \emph{maximum matching}, \emph{dominating set}, \emph{vertex cover}, and \emph{densest subgraph}.
% , and \emph{optimal transport} problems. 
% \st{Next, we explain how we can form the corresponding LP formulations and their formulations in matrix form for these graph applications.} \caleb{I think we can remove since the first paragraph of this section explains this already}
% \caleb{Next, we define these problems' corresponding IP or LP formulations first and then their  LP relaxation in matrix form.}

\noindent
\textbf{Maximum Matching (\matching).} A matching is a subset $F \subset E$ of edges such that no vertex is 
incident to more than one edge in $F$. %An example of matching is given in Figure~\ref{fig:matching_vis}. The set of edges in the matching are marked with thick red lines. The goal is to find a maximum cardinality matching.
We can write this optimization problem as the IP,
% This optimization problem can be written as 
\begin{equation} \label{eq:matching}  
\begin{split}
	\max  \sum\limits_{e \in E} x_e  \ 
	\text{s.t.} & \ \sum\limits_{e \in \mathrm{inc}(v)} x_e \leq 1 \ \forall v \in V \\
	& \ x_e  \in \{0,1\} \ \forall e \in E.
\end{split}
\end{equation}
% Here, the variables $x_i$ act as boolean variables that indicate if edge $i$ is in $F$. and the constraints are over the vertices.
% When the input graph is a bipartite graph, we call the problem 
% \textbf{maximum bipartite matching}.
The $x_e$ variables indicate whether edge $e$ is in set $F$. The constraints are defined over the vertices of the graph such that at most one edge ($x_e$) in the incident edges of a vertex $v$ ($\mathrm{inc}(v)$) can be selected for matching set $F$.
% and the constraints are over the vertices.
When the input graph is a bipartite graph, we call this problem \textbf{maximum bipartite matching} (\bimatching).

% \begin{figure}[h!]
\begin{figure}[t]
    \captionsetup[subfloat]{farskip=0pt,captionskip=0pt}
 \centering
    \subfloat[\matching]{\includegraphics[width=1.8cm]{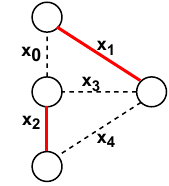}\label{fig:matching_vis}}
   ~
    % \captionsetup[subfloat]{farskip=0pt,captionskip=0pt}
    \subfloat[\domset]{\includegraphics[width=1.6cm]{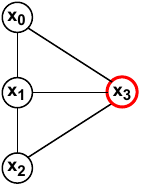}\label{fig:domset_vis}}
   ~
%   \captionsetup[subfloat]{farskip=0pt,captionskip=0pt}
    \subfloat[\vcover]{\includegraphics[width=1.6cm]{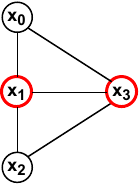}\label{fig:vcover_vis}}
   ~
%   \captionsetup[subfloat]{farskip=0pt,captionskip=0pt}
    \subfloat[\densest]{\includegraphics[width=1.85cm]{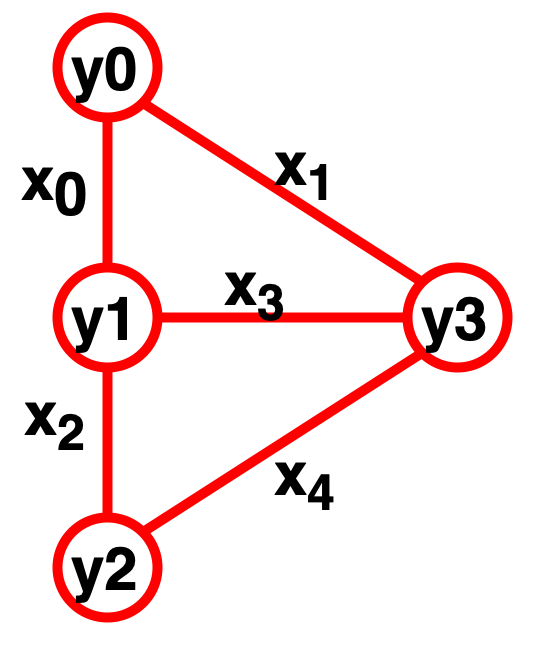}\label{fig:densest_vis}}
   
\centering
\caption{Four graph problems run on the same graph. Variables of LP with a nonzero value are highlighted in red. An example of matching is given in Figure~\ref{fig:matching_vis}. The set of edges in the matching are marked with thick red lines.
% Figure (a) shows the maximum matching,
% which is the set edges represented by solid red lines (and dashed
% lines are the other edges). Figure (b) shows the smallest dominating 
% set, which is the single vertex marked in red. Figure (c) shows the
% smallest edge cover, which is the set of vertices marked in red. 
% Figure (d) shows the densest subgraph, denoted by the subgraph the
% single red vertex and three solid red lines. The density of this
% subgraph is 3.
}
\label{fig:ProbEx}
\end{figure}

Given the vertex-edge incidence matrix $\B{M}$ of the graph, we can write the LP relaxation for
maximum matching as the pure packing LP,
% {
\begin{equation} 
	\max  \ \langle \mathbb{1}, \B x \rangle
	\text{ s.t. } \B{Mx} \leq \mathbb{1},  
	\B x \geq \mathbb{0}, \ \B x \in \mathbf{R}^m.
\end{equation}
% $\B{M}$ is an vertex-edge incidence matrix. 
% }
It is well-known that this LP relaxation has no integrality gap for \bimatching while for general graphs there is an integrality gap of $2/3$; there is an exponential sized exact LP relaxation for general graph matching but we do not consider it here. 
% \chandra{Reworded and added some details.}

\noindent
\textbf{Dominating Set (\domset).} 
% In the \textbf{dominating set} problem, the goal is
% to find the smallest subset of vertices $S \subset V$ such that every vertex in the graph is either in $S$ or is a neighbor of a vertex in $S$. 
Dominating set is the problem of finding the smallest subset of vertices $S \subseteq V$ such that every vertex in the graph is either in $S$ or is a neighbor of a vertex in $S$. %Figure~\ref{fig:domset_vis} shows an example \domset solution of our example graph. In this example, the dominating set of given graph has only one vertex (marked with red), which has an edge to all other vertices.
% Its formulation is
We can formulate the \domset problem as the IP,
\begin{equation} 
\begin{split}
	\min \ \sum\limits_{v \in V} x_v \ 
	\text{ s.t. } & \ x_v + \sum\limits_{u \in N(v)} x_u \geq 1, \ \forall v \in V \\
	& \ x_v  \in \{0,1\} \ \forall v \in V.
\end{split}
\end{equation}
The variable $x_v$ indicates whether vertex $v$ is in set $S$ or not. The constraints are defined over
the vertices such that either vertex $v$ itself or one of its neighbors is in the set $S$. The LP
relaxation is a pure covering LP,
\begin{equation} 
\begin{split}
	\min \ \langle \mathbb{1}, \B x \rangle
	\text{ s.t. } (\B{I}+\B{A})\B x \geq \mathbb{1}, 
	\B x  \geq \mathbb{0}, \ \B x \in \mathbf{R}^n,
\end{split}
\end{equation}
where $\B{I}$ is the identity matrix. % \chandra{Reworded a bit.}

% We now introduce the \textbf{vertex cover} problem. 

\noindent\textbf{Vertex Cover (\vcover).} In the vertex cover problem the goal is to find the smallest
subset of vertices $S \subseteq V$ such that every edge has one of its endpoints
in $S$ (hence $S$ covers all the edges). %Figure~\ref{fig:vcover_vis} shows an example four node graph and its vertex cover.
A simple IP formulation is
\begin{equation} 
\begin{split}
	\min \ \sum\limits_{v \in V} x_v \ 
	\text{s.t.} & \  x_u + x_v \geq 1 \ \forall (u,v) \in E \\
	& \ x_v  \in \{0,1\} \ \forall v \in V.
\end{split}
\end{equation}
Here, the $x_v$ variables determines whether a vertex $v$ is in set $S$. These variables are defined over the vertices. There is one constraint per edge. The LP relaxation is a pure covering LP,
% Its matrix formulation is
\begin{equation}  
\begin{split}
	\min \ \langle \mathbb{1}, \B x \rangle
	\text{ s.t. }  \B{M}^\mathsf T x\B  \geq \mathbb{1}, 
	\B x \geq \mathbb{0}, \B x \in \mathbf{R}^n,
\end{split}
\end{equation}
where $\B{M}^T$ is the transpose of the vertex-edge incidence matrix of the graph.
% \chandra{Reworded a bit.}

\noindent\textbf{Densest Subgraph (\densest).} 
% In the \textbf{densest subgraph} problem,
Densest subgraph finds a subgraph $S \subset G$ that maximizes the edge to vertex
count ratio, i.e., $|E(S)|/|S|$. %Figure~\ref{fig:densest_vis} shows the densest subgraph of the example graph, where $S$ is the whole graph 
%with $|E(S)|/|S|=5/4$.
The LP, as formulated
in~\cite{charikar2000greedy}, is
\begin{equation} \label{eq:densetSubgraphLP}
\begin{split}
	\max \ \sum\limits_{e \in E} x_e  \ 
	\text{s.t.} & \ x_e \leq y_u, x_e \leq y_v \ \ \forall e = (u,v) \in E \\
	% & \ x_i \leq y_v \ \forall e_i = (u,v) \in E \\
	& \ \sum\limits_{v \in V} y_v \leq 1 \\
	& \ x_e,y_v  \geq 0 \ \forall v \in V, \forall e \in E.
\end{split}
\end{equation}
The variables $x_e$ represent the edges and the variables $y_v$ represent
the vertices, which are no longer binary. 
% Moreover, the structure of densest subgraph is different than previous problems.
% require additional
% definitions. 
Since this problem is not a positive LP, we consider its dual~\cite{boob2020flowless}, 
% as a positive LP. In this case, we seek to find a minimum feasible solution $D$ with the following positive LP is feasible,
\begin{equation} \label{eq:dualDS} 
\begin{split}
	\min D \text{ s.t.}& \ z_{u,e} + z_{v,e} \geq 1 \ \forall e = (u,v) \in E \\
	& \ \sum\limits_{e \in \text{inc}(v)} z_{v,e} \leq D \ \forall v \in V \\
	&\ z_{v,e} \geq 0 \ \forall v \in V, e \in \mathrm{inc}(v).
\end{split}
\end{equation}
While~\eqref{eq:dualDS} is still not a positive LP, we can convert it to a
mixed packing and covering LP by fixing $D$ to be a constant and 
treating~\eqref{eq:dualDS} as a feasibility problem instead. We find an
approximate minimum value to~\eqref{eq:dualDS} via binary search for $D$ and solving the feasibility LP for each choice. % \chandra{Reworded a bit.}

% In this case, the problem becomes a feasibility problem instead of a minimization problem. With this formulation, we seek to find a minimum solution $D$ such that the positive LP is satisfied. To solve this problem, we will search for the smallest $D$
% where the above problem is feasible by conducting an exponential search to
% find upper and lower bounds on the minimum $D$ followed by binary search.

Unlike previous LPs, the variables
$z_{v,e}$ represent a vertex-edge pair instead of vertices or edges,
so we require new constraint matrices.
% cannot use the adjacency or vertex-edge incidence matrix. 
% So, we cannot represent it in matrix form by directly using the adjacency or vertex-incidence matrices. 
% However, the constraint matrices can still be obtained with simple transformations on the vertex-edge incidence and identity matrices.
% \serif{@caleb can we write these definitions referring to the variables}
Let $\B{I}$ be an identity matrix of size $m$, and let the 
function $\mathrm{interweave}$ take two equally-sized matrices and put the 
first and second matrices' first columns as the first two columns of the 
combined matrix, then their second columns as the next pair of columns, and so 
on. 
% \st{We call the \emph{interweaved identity matrix} as $\B{W}$. Specifically, matrix $\B{W}$ is formed as follows:}
We call the resulting matrix $\B{W}$ the \emph{interweaved identity matrix},
where
\begin{equation} \label{eq:interweaved} 
    \B{W}_{e,2e}, \B{W}_{e,2e+1} = 1,\  \forall e \in E, \B{W} \in \{0,1\}^{n \times 2m}
\end{equation}

% \serif{need to get rid of this oddshift}Moreover, the function $\mathrm{oddshift}$ takes in a matrix with two non-zeros 
% (e.g., the vertex-edge incidence matrix) in each column. 
In order to model the vertex-edge pair variables, we can form a matrix called vertex-edge pairs matrix. Vertex-edge pairs matrix will have a column for each vertex $v$ and edge ($u,v$).
% Consider the vertex-incidence matrix which has 2 nonzeros in each column.
% To form vertex-edge pair matrix, we can take the bottom non-zero in each column and move it into a new column directly 
% to the right of the current column, creating 2 columns for every edge of the original graph. 
Specifically, we can form the vertex-edge pair matrix, $\B{O} \in \{0,1\}^{n \times 2m}$, from a graph $G$ as follows:
\begin{equation} \label{eq:oddshifted_incidence}
\begin{split}
    \B{O}_{u,2e+b} = \begin{cases} 
    % 1 &: e=(u,v) \in E \\
    1 &: b=0, e=(u,v) \in E \\
    1 &: b=1, e=(v,u) \in E \\
    0 &: \text{otherwise}
    \end{cases}
    % \B{O}  \in \{0,1\}^{n \times 2m}.
\end{split}
\end{equation}
Then the feasibility variant of the dual to densest subgraph is
% \begin{equation} 
% \begin{split}
% 	\min D 
% 	\text{ s.t.} &\ \mathrm{interweave}(\B{I}, \B{I})z \geq \mathbb{1} \\
% 	&\ \mathrm{oddshift}(\B{M})\B z \leq D \cdot \mathbb{1}\\
% 	&\ \B z \geq \mathbb{0}, \B z \in \mathbb{R}^{2|E|}.
% \end{split}
% \end{equation}

\begin{equation} \label{eq:densest_linalg}
% \begin{split}
	\exists \B z \in \mathbb{R}^{2m}
	\text{ s.t.} \ \B{W}z \geq \mathbb{1},
	\ \B{O}\B z \leq D \cdot \mathbb{1},
	\ \B z \geq \mathbb{0}.
% \end{split}
\end{equation}

\section{
\sout{Heuristics to Improve Convergence in Practice} MWU with Line Search} \label{sec:AlgImprove}
% We now introduce algorithmic modifications to Algorithm~\ref{alg:exp1} that improve its performance in practice, which we review in \ref{sec:sec:step_strat_exp}. We write it in Algorithm~\ref{alg:exp2} and mark the modifications in red. Similar to Algorithm~\ref{alg:exp1}, we can adapt this algorithm to solve pure packing (e.g., maximum matching) and pure covering (e.g., dominating set, vertex cover) LPs more efficiently, as described in Section~\ref{sec:mwu_alg_sec}.
We propose two methods for finding a step size in \texttt{StepSize}
(Line~\ref{line:find_step_size}). To motivate these methods, we
cast them as 1-D optimization problems, similar to line search methods for
(gradient) descent methods~\cite{nocedal2006numerical}. However, instead of
finding a step size that minimizes the objective
function~\cite{nocedal2006numerical}, we take the largest step that ensures a
local invariance condition is satisfied.

% \serif{is it possible to add comments to lines~\ref{line:smax},~\ref{line:smin} -- gradient calculation, and line~\ref{line:direction} -- new direction, line~\ref{line:find_step_size} -- Search for step size?}
\begin{algorithm}[h]
\caption{Multi-Update MWU Method with step size search for Mixed Packing and Covering LPs} \label{alg:alg2}
\footnotesize
\begin{algorithmic}[1]
% \small
\Procedure{MWU\_STEPSIZE\_SEARCH}{$\B{P} \in \mathbf{R}_+^{m_P \times n}$, $\B{C} \in \mathbf{R}_+^{m_C \times n}$, $\epsilon$ }
  \State{Initialize $\eta$, and $x_i$ as in Algorithm~\ref{alg:exp1}} \label{line:init2}
  % \State $\text{num\_iter} \gets 0$
  % \State $\eta \gets 10 \log(m_p+m_c)/\epsilon$
  % \State $x_i \gets \frac{\epsilon}{n \left\Vert \B{P}_{:,i}\right\Vert_\infty} \ \forall i \in [n]$
  \State {$\B y \gets \B{Px}$, $\B z \gets \B{Cx}$} \label{lst:line:cons1}
  % \State If $M$ not given, set $M \gets \mathrm{GuessM}(\B{P},x_0,\epsilon)$
  % \While{constraints not approximately satisfied and {\color{red} $\max(\B y) < 1$}}  \label{line:newstop}
  \While{constraints not approximately satisfied and $\B C \ne \varnothing$} 
    \State {$\B g \gets \B{P}^\mathsf T \nabla \mathrm{smax}_\eta(\B{y})$} \label{line:smax} \Comment{Packing gradient}
    \State {$\B h \gets \B{C}^\mathsf T \nabla \mathrm{smin}_\eta(\B{z})$} \label{line:smin}\Comment{Covering gradient}
    \State{$d_i \gets \frac{1}{2\eta}\mathrm{max} \{0,1-\frac{g_i}{h_i}\} \cdot x_i \ \forall i$} \label{line:direction} \Comment{New step direction}
    % \State $\B{d}^\star \gets \max(\B d)$$
    \If{$\max(\B{d}) = 0$} \label{line:get_maxd}
        \State{Return ``INFEASIBLE''} \label{line:infeasible}
    \EndIf
    % \State $\B d\_\text{sum} \gets \mathrm{sum}(d)$ 
    \State{$\B{d}^{(y)} \gets \B{Pd}$, $\B{d}^{(z)} \gets \B{Cd}$} \label{line:newcons}
    \State {\color{red} $\alpha \gets \operatorname{StepSize}(\B d, \B y, \B z, \B{d}^{(y)}, \B{d}^{(z)}, \eta)$}  \label{line:find_step_size} \Comment{Step size search}
    \If{$\alpha < 1$}
        \State{Return ``INFEASIBLE''} \label{line:term_step_size} % \label{Alg2:badstepsize}
    \EndIf
    \State{$\B x \gets \B x + {\color{red} \alpha} \cdot \B d$} \label{line:newupdate}
    \State{$\B y \gets \B y + {\color{red} \alpha} \cdot \B{d}^{(y)}$,  $\B z \gets \B z + {\color{red} \alpha} \cdot \B{d}^{(z)}$} \label{line:cons2}
    % \State {\color{red} \slash \slash \  $\B C \gets \{ c_i \ : \ c_i^Tx < 1 \}$} \label{line:nodrop}
    \State {$\B C \gets \{ c_i \ : \ c_i^Tx < 1 \}$} \label{line:nodrop}
    % \State{$\text{num\_iter} \gets \text{num\_iter} + 1$}
  \EndWhile
  % \State $x \gets \mathrm{LastStep}(\B{P},x)$ \Comment If $\B{P}x < \mathbb{1}$, takes one last step
  \State{\Return $x$}
\EndProcedure
\end{algorithmic}
\label{alg:exp2}
\end{algorithm}

% \serif{We need to explicitly say Line when we are referring to the lines}
% We initialize all the variables the same as before (Line~\ref{line:init2}. In addition, 
% \serif{We need to check line numbers?Seems out of order}
We store the packing and covering constraints 
$\B y = \B{Px}$ and $\B z = \B{Cx}$ as well as
$\B{d}^{(y)} = \B{Pd}$ and $\B{d}^{(z)} = \B{Cd}$ (Line~\ref{lst:line:cons1} and~\ref{line:newcons}) to minimize the number of sparse matrix-vector products, or SpMVs.
% We also replace the stopping criteria to the violation of a packing constraint (Line~\ref{line:newstop}) since we do not drop covering constraints (Line~\ref{line:nodrop}).  Although removing this step can make the iteration count of MWU dependent on the values of $\B C$, we find it did not impact the convergence on the problems we tested.
The step direction $\B d$ is unchanged (Line~\ref{line:direction}). While the algorithm drops satisfied constraints (Line~\ref{line:nodrop}), in practice we keep satisfied constraints since this simplifies the implementation and we did not find it impacts the convergence on the problems we tested.

The sub-routine \texttt{StepSize} (Line~\ref{line:find_step_size}) takes the step vector $\B d$ and constraint vectors, and returns a step size $\alpha > 0$. We call 
this modification \emph{step size search} and design algorithms for it in the next subsection. When 
$\alpha < 1$, we report that finding a solution is infeasible, because otherwise a step size of
$\alpha = 1$ is always possible % in Algorithm~\ref{alg:exp1} 
due to the theoretical analysis of Mahoney et. al.~\cite{mahoney2016approximating}. Therefore, we call a step size $\alpha=1$ found without step size search the \emph{standard step size}. Assuming $\alpha \geq 1$, we
scale the step direction $\B d$ by $\alpha$ and add it to 
$\B x$ (Line~\ref{line:newupdate}). Afterwards, we update
$\B y = \B{Px}$ and 
$\B z= \B{Cx}$ without SpMVs (Line~\ref{line:cons2}). Note that $\alpha$ may be large enough so that
$\B{Px} = \B y + \alpha \cdot \B{d}^{(z)} \geq 1$, in which case we terminate MWU.

\subsection{Line Search as a Constrained Optimization Problem}
In this section, we consider algorithms for finding a step size for \texttt{StepSize}. 
% which is scalar $\alpha$ multiplied by the step direction $\B d$ (Line~\ref{line:newupdate} and~\ref{line:cons2}). 
When selecting $\alpha$, we want it to be
sufficiently large to accelerate MWU convergence while ensuring that we recover a feasible solution to~\eqref{eq:NormFeasibleLP}.
%\textcolor{red}{The method derived here shares many similarities with the binary search step size method derived by Kent Quarund \cite{quanrud2019thesis}, with the exception of some small differences, which will be described later in the section}.

%\textcolor{orange}{Recall from Section~\ref{sec:theory} that the step $\B d$ in MWU satisfies the bang-for-buck inequality~\eqref{eq:OGB4B}. From our initial experiments, we observed that this ratio is often much larger than 1 whereas the individual values of $\B d$ are relatively small. In other words, MWU is overly cautious. By scaling the step $\B d$ with the step size $\alpha > 0$, MWU will take larger and more aggressive steps. When selecting a step size $\alpha$, we want $\alpha \gg 1$ so that MWU rapidly}

First, we consider the case of the mixed packing covering LP. One of the qualities of Mahoney et. al's algorithm is that the algorithm reaches a $(1+\epsilon)$-approximation when the difference in a potential function $f(x) = \frac{1}{\eta} ( \mathrm{smax}_{\eta}(\B{Px}) - \mathrm{smin}_{\eta}(\B{Cx}) )$ becomes sufficiently small \cite{mahoney2016approximating}. Each step that their algorithm takes is non-increasing on $f(x)$.

%\textcolor{purple}{
%\begin{equation}
%\begin{split}
%\Psi(\alpha) \le \alpha \langle \B g, (1+\B d) \odot \B d \odot \B x \rangle \\
%\Phi(\alpha) \ge \alpha \langle \B h, (1-\B d) \odot \B d \odot \B x \rangle
%\end{split}
%\end{equation}}

We can show that in order for the potential function to be non-increasing, $f(x^{(t+1)}) - f(x^{(t)}) = \Psi(\alpha) - \Phi(\alpha) \le 0$
where,
\begin{align*} 
    \Phi(\alpha) &= \mathrm{smin}_{\eta}(\B{C}(\B x + \alpha \cdot \B d)) - \mathrm{smin}_{\eta}(\B{Cx}) \\
    \Psi(\alpha) &= \mathrm{smax}_{\eta}(\B{P}(\B x + \alpha \cdot \B d)) - \mathrm{smax}_{\eta}(\B{Px}).
\end{align*}

%\textcolor{purple}{Taking $d_i = \frac{1}{2\eta}\mathrm{max} \{0,1-\frac{g_i}{h_i}\} \cdot x_i$, we can show that all the steps in the original algorithm satisfy the invariant $f(\alpha) =  \Phi(\alpha) / \Psi(\alpha) \leq 1$.}
This gives us an equivalent invariant $f(\alpha) =  \Phi(\alpha) / \Psi(\alpha) \geq 1$. Therefore, if we find a step size $\alpha$ for which this invariant holds, then for this $\alpha$ we can still say $f(x)$ is non-increasing, and furthermore, if we can reach a point where $\B{Px} \le (1+\epsilon)\mathbb{1}$ and  $\B{Cx} \ge \mathbb{1}$ then we have converged to a feasible solution.
Hence, our method is to find the largest step size $\alpha>0$ such that the ``bang-for-buck'' value is at least one, or
%We formalize the latter condition with the bang-for-buck inequality
\begin{equation} \label{eq:b4b_conv} % \footnotesize
f(\alpha) =  \Phi(\alpha) / \Psi(\alpha) \geq 1,
\end{equation}
For pure packing and pure covering problems we instead have these invariants, respectively:
\begin{equation}
\begin{split}
\langle \mathbb{1}, \alpha \B d \rangle / \Psi(\alpha) \ge 1 \\
\langle \mathbb{1}, \alpha \B d \rangle / \Phi(\alpha) \le 1 \\
\end{split}
\end{equation}

We now show MWU with line search maintains the same theoretical properties as MWU with the standard step size~\cite{mahoney2016approximating}. Recall $\B x \in \mathbf{R}^n$ and $m$ is the number of rows in the matrices $\B{P}$ and $\B{C}$.
\begin{theorem}
    MWU with line search (Algorithm~\ref{alg:exp2}) either returns an  $(1+\epsilon)$-relative approximate solution, i.e., an $x \geq \mathbb{0}$ such that $\B{Px} \leq (1+\epsilon) \mathbb{1}$ and $\B{Cx} \geq \mathbb{1}$, or correctly reports the LP is infeasible. The number of iterations is at most $\tilde{O}(\epsilon^{-3})$, where $\tilde{O}$ hides polylogarthmic dependence on $n$, $m$, and $\epsilon$. 
\end{theorem}
The proof is similar to the one shown in~\cite{mahoney2016approximating}, which implicitly sets the step size to $\alpha = 1$. There will be two main differences in the convergence proof, which we highlight here. First, the proof of correctness in ~\cite{mahoney2016approximating} shows the potential function, defined as $\mathrm{smax}_\eta(\B{Px}) - \mathrm{smin}_{\eta}(\B{Cx})$, is monotonically decreasing by taking a first-order approximation of smooth max and min. On the other hand, our bang-for-buck invariance~\eqref{eq:b4b_conv} explicitly ensures this monotonicity property. Second, the argument in~\cite{mahoney2016approximating} upper bounds the number of MWU iterations by lower bounding the values in step direction vector $\B{d}$. Since line search only increases $\B{d}$ because $\alpha \cdot \B{d} \geq \B{d}$, then line search can only decrease the number of MWU iterations.

% \sout{This inequality enforces the covering constraints to grow approximately no slower than the packing constraints.}
% We leverage the monotonicity of $f(\alpha)$ to design efficient line-search algorithms. \textcolor{purple}{Note that this step size search has no theoretic guarantee on reducing the number of iterations, however, it works efficiently in practice.}
% \textcolor{blue}{Caleb: including this sentence since it provides a better transition between the two results than the previosu sentence.} 
While line search can decrease the number of iterations (in fact, quite significantly in our experiments), finding a step size increases the work per iteration. In the following lemma, we leverage the monotonicitiy of $f(\alpha)$ to design efficient line search algorithms.

% We observe in practice that $\alpha \gg 1$, especially during the beginning of
% MWU. As we show in Section~blah, the incorporation of line search can drastically
% reduce the number of iterations. Moreover, the cost of line search is 
% often only a small fraction of the total run time.

%We also want the step size search to be efficient.
%One approach is to use an
%efficient solver for the constrained optimization problem
%\begin{equation} \label{eq:StepOptProb} \footnotesize
%\begin{split}
%    \max \alpha \text{ s.t. }
%    f(\alpha) &= \frac{\Phi(\alpha)}{\Psi(\alpha)} \geq 1 \\
%    \text{where }\Psi(\alpha) &= 
%	\mathrm{smax}_\eta(\B{P}(\B x + \alpha \cdot \B d)) - \mathrm{smax}_\eta(\B{P}\B x ) \\
%    \Phi(\alpha) &= 
%	\mathrm{smin}_\eta(\B{C}(\B x + \alpha \cdot \B d)) - \mathrm{smin}_\eta(\B{C}\B x)
%\end{split}
%\end{equation}
%
%
%We will show there are simple and fast methods to solve~\eqref{eq:StepOptProb}.
%First,~\eqref{eq:StepOptProb} is a 1D optimization problem. 
%Second, we have the following crucial property:
\begin{proposition} \label{Prop.MonoInc}
    $f$ is monotonically decreasing for $\alpha \in \mathbf{R}_+$.
\end{proposition}

\begin{proof}
We show that as $\alpha$ increases $\Psi(\alpha)/\alpha$ is increasing while $\Phi(\alpha)/\alpha$ is decreasing, hence
\[f(\alpha) = (\Phi(\alpha)/\alpha)/(\Psi(\alpha)/\alpha)\] is decreasing.
Note that $\Psi$ is convex since $\mathrm{smax}_\eta$ is convex, and $\Phi$ is concave since $\mathrm{smin}_\eta$ is concave.
Since $\Psi$ is convex, $\Psi(\alpha) \leq \Psi(0) + \alpha \Psi'(\alpha) = \alpha \Psi'(\alpha)$.
Hence, we can show that $\Psi(\alpha)/\alpha$ is increasing, since
\[(\Psi(\alpha)/\alpha)' = \frac{1}{\alpha}(\Psi'(\alpha) - \Psi(\alpha)/\alpha) \geq 0.\]
Analogously, since $\Phi$ is concave, the inequalities above are reversed, and so $\Phi(\alpha)/\alpha$ must be strictly decreasing in $\alpha$.
\end{proof}

\subsection{Implementing Line Search}
To approximate the maximum step size $\alpha^*$ satisfying~\eqref{eq:b4b_conv}, we first perform exponential search to find an integer $p$ where $f(2^p) \geq 1$ and $f(2^{p+1}) < 1$. By Proposition~\ref{Prop.MonoInc}, 
$\alpha^* \in [2^p,2^{p+1})$. Next, we run binary search
starting with a lower and upper bound of $l=2^p$ and
$u=2^{p+1}$, and update the lower and upper bounds so that $l$ (resp. $u$) 
is the largest (smallest) value such that $f(l) \geq 1$ ($f(u) < 1$). 
Once we find an $\epsilon$-relative step size, or when 
$\frac{u-l}{l} \leq \epsilon$, we return $l$. We formalize the aforementioned procedure in Algorithm~\ref{alg:binary_step}.

When $f(\alpha) \geq 1$ and 
$\max\big(\B{C}( \B x + \alpha \cdot \B{d})\big) \geq 1$ (Line~\ref{line:early_stop_bin}), this means the
step size can lead MWU to completion, so we
immediately return $\alpha$. Moreover, 
binary search makes use of $\B y = \B{Px}$, $\B z = \B{Cx}$, $\B{d}^{(y)} = \B{Pd}$, and
$\B{d}^{(z)} = \B{Cd}$, where $\B x$ is our current solution and $\B d$ is
the computed MWU update direction, to avoid additional SpMVs. 

A similar line search, which is a coordinate binary search, was proposed
in~\cite[Section 2.8]{quanrud2019thesis}, but there are some important
differences compared to our binary search. First, rather than taking a step in
the full gradient direction $\B d \in \mathbb{R}^n$, the coordinate binary
search updates sequentially in each of the $n$ indices. Thus, the binary search
can have a critical path of length up to $n$ and is therefore not parallel.
Second, the coordinate binary search replaces the difference in the smooth min
and smooth max from $\Phi$ and $\Psi$  with their respective first-order
approximations, which incurs approximation errors in~\eqref{eq:b4b_conv} and can
lead to more conservative step sizes. By conducting line search in the full
gradient $\B d$ and using the exact difference for $\Phi$ and $\Psi$, our
proposed binary search improves upon the previous line search in non-trivial
ways and can take more aggressive (i.e., larger) step sizes while maintaining
feasibility.

\begin{algorithm}[h]
\caption{Finding a step size via binary search}
\footnotesize
\begin{algorithmic}[1]
\Procedure{BinSearch}{$\{\B y, \B{d}^{(y)} \}\in \mathbf{R}^{m_p}_{\geq 0}$, 
$\{ \B z, \B{d}^{(z)} \} \in \mathbf{R}^{m_c}_{\geq 0}$, $\epsilon$}       
  \State $\alpha \gets 1$ 
  \While{$f(\alpha) \geq 1$} \Comment{Exponential search. See~\eqref{eq:b4b_conv}} % \land \mathrm{max}(y + \alpha \cdot dy) \leq 1$\big)} \Comment{See~\ref{eq:fun1}}
    \If{ $\mathrm{min}(\B z + \alpha \cdot \B{d}^{(z)}) \geq 1$ } \label{line:early_stop_bin}
      \State \Return $\alpha$ \Comment{Return early if constraints are satisfied}
    \EndIf
    \State $\alpha \gets 2 \cdot \alpha$
  \EndWhile
  \State $lb, ub \gets \alpha/2, \alpha$ \label{line:bin_bnds}
  \While{$ub-lb > (1-\epsilon) lb$} \Comment{Binary search}
    \State $\beta \gets \mathrm{avg}(lb,ub)$
  	\If{$f(\beta) \geq 1$} \Comment{See~\eqref{eq:b4b_conv}}
		\State $lb \gets \beta$
  	\Else
    	\State $ub \gets \beta$
	\EndIf
  \EndWhile
  \State $\alpha \gets lb \text{ and } \Return \ \alpha$
\label{roys_loop_binary}
\EndProcedure
\end{algorithmic}
\label{alg:binary_step}
\end{algorithm}

% \subsection{Newton's Method}
We can derive another line search method by % viewing~\eqref{eq:b4b_conv} as a root-finding problem,
% \begin{equation*} % \label{eq:root_find}
%     g(\alpha^*) = f(\alpha^*)-1=0.
% \end{equation*} 
% We can solve this with 
using Newton's method, which has the update,
\begin{equation*} % \label{eq:newton_update} 
    \alpha_{k+1} = \alpha_{k} - g(\alpha_{k})/g'(\alpha_{k}), \text{ where } g(\alpha) = f(\alpha)-1.
\end{equation*}
% The gradient has a simple closed-form expression,
% \begin{equation} \label{eq:grad_f}
%     g'(\alpha_k) = \frac{\Psi(\alpha_k) \cdot \Phi'(\alpha_k) - \Phi(\alpha_k) \cdot \Psi'(\alpha_k)}{\Psi(\alpha_k)^2},
%     % \frac{ \langle \mathbb{1}, d \rangle }{ \Psi(\alpha)^2 } \cdot \Big( \Psi(\alpha) - \alpha \cdot d^Tg^{(k+1)}(\alpha) \Big)
% \end{equation}
% where
% \begin{equation} \label{eq:grad_slim}
% \begin{split}
%     % g^{(k+1)}(\alpha) &= \frac{ \mathrm{exp}\big( \eta \B{P}(x + \alpha \cdot d) \big) }{ \langle \mathbb{1},  \eta \B{P}(x + \alpha \cdot d) \rangle }.
%     \Psi'(\alpha_k) &= \frac{\langle \B{Pd}, \mathrm{exp}(\eta (\B{Px} + \alpha_k \cdot \B{Pd})) \rangle }{\langle \mathbb{1}, \mathrm{exp}(\eta (\B{Pd} + \alpha_k \cdot \B{Pd})) \rangle}, \\
%     \Phi'(\alpha_k) &= \frac{\langle \B{Cd}, \mathrm{exp}(-\eta (\B{Cx} + \alpha_k \cdot \B{Cd})) \rangle }{\langle \mathbb{1}, \mathrm{exp}(-\eta (\B{Cx} + \alpha_k \cdot \B{Cd})) \rangle}.
% \end{split}
% \end{equation}
Because Newton's method converges when its solution is in the neighborhood of the optimal solution, we require estimates of $\alpha^*$ to ensure convergence. We do so via a warm start for Newton's search, where we set our initial $\alpha_0$ to the previous optimal step size, if available, or use exponential search. The reason for the former strategy is we observed in our tests that the optimal step size between two MWU iteration are relatively close. Finally, we note that once Newton's method converges to some solution, it may not strictly satisfy~\eqref{eq:b4b_conv}. Thus, we multiplicatively decrease the solution by a factor of $(1-\epsilon)^p$ for some integer $p$ ($p$ is typically small) until~\eqref{eq:b4b_conv} is satisfied.

\section{Software Optimizations and Parallelization}
\label{sec:optimizations}
% \serif{this section is too long, should be compressed}
% \serif{@Caleb: please go over this section for checking both the language and content}
% the convergence of 
% , i.e.
We now describe the details of our implementation and parallelization of linear algebra operations within MWU.
% However, a holistic approach should also consider the time it takes to execute a single iteration. 
% In MWU specification shown in Algorithm~\ref{alg:exp1}, we identify 
% removed
% We identify two primitives consuming the majority of the execution cycles in MWU. These are embarrassingly parallel but computationally expensive vector operations and Sparse Matrix-Vector Multiply (SpMV) operations.
% % \serif{needs more details}
% % \sout{In order t}
% To achieve maximum performance, we aim to accelerate 
% % target acceleration of 
% both of these primitives. 
% end removed
% On the other hand, 
% \sout{for SpMV,} 
To efficiently perform sparse matrix-vector products with matrices introduced in Section~\ref{sec:MWUProbs}, such as the vertex-edge adjacency matrix, we leverage implicit representations derived from a standard sparse vertex-vertex adjacency matrix data structure.
%We also optimize the SpMV performance by selecting a favorable matrix representation 
% \sout{by borrowing from previous works} 
%and providing efficient parallel implementations of SpMV for commonly 
% observed 
%encountered matrices such as the adjacency, vertex-edge incidence, and 
We accelerate vector operations with loop fusion and vectorization. 
\subsection{Shared-Memory Optimizations}

\mina{We design implicit SpMVs for matrices that arise in graph-based positive LPs such as vertex-edge incidence matrices. We adapt previous fusion techniques for the needs of MWU framework~\cite{fusion_ics98}.}

\subsubsection{Choice of Matrix Format} 
% In MWU framework, as in many LP solvers, both the original and transpose of the input matrix is needed. 
% Each iteration of MWU requires an SpMV with the constraint matrices and their transpose.  % For this reason, we choose to use the 

%In this section, we examine design choices
%that offer the best performance for SpMVs with these \emph{implicit matrix representations}.

To efficiently traverse the non-zeros in the adjacency matrix during SpMVs, we use the Compressed Sparse Blocks (CSB) format~\cite{bulucc2009parallel}, which can achieve good cache locality for both SpMVs of the matrix and its transpose.
% to represent the constraint matrices because
% it does not prefer columns over rows or rows or columns. 
% Therefore, it provides similar performance for SpMV of the matrix and its transpose. 
% As a result, CSB is a good fit for our purposes. 
% As we will see later, it is also suitable for implicit SpMV operations.
% implicitly executing SpMV on specialized matrices derived from the adjacency matrix of a graph. 

CSB divides the matrix into two-dimensional $r\times k$ tiles. Each tile is represented as a list of tuples, where each tuple stores the non-zeros in column major order in coordinate (COO) format. 
The group of tiles that belong to $r$ consecutive rows is called a \emph{row-block} while the group of tiles that belong to $c$ consecutive columns is called a \emph{column-block}. 
% All tiles of consecutive $r$ rows are called a \emph{row-block} while all tiles that belong to the $k$ consecutive columns are called  \emph{column-block}. 
In our implementation, we store the tiles 
% of a given matrix 
in row-major order. 
% \caleb{So $k$ and $c$ are separate variables.}\serif{yes, do you think it requires clarification?} No, just wanted to make sure
% \caleb{Are $r$, $k$, and $c$ up to the user to define? Also, why is $r$ used twice but $k$ and $c$ each once?}\serif{yes}

% \begin{comment}
\begin{figure}[h]
    \centering
    \includegraphics[scale=0.4]{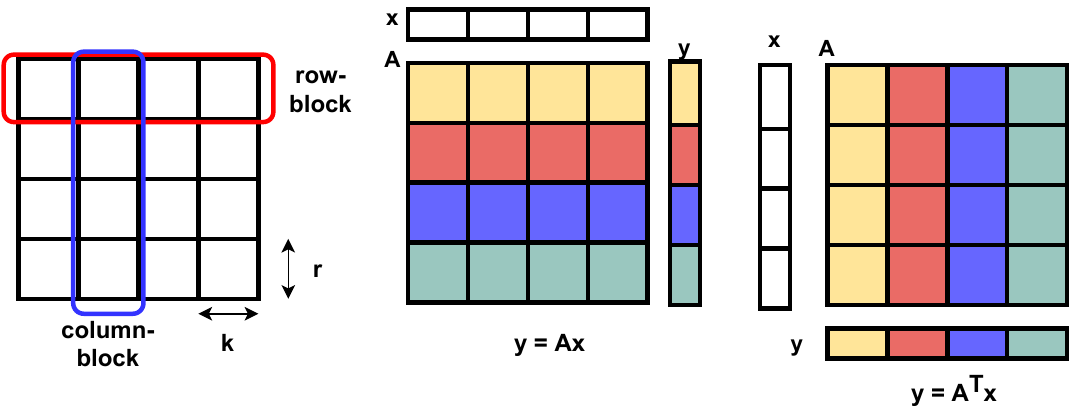}
    \caption{CSB representation and its SpMV operation.}
    \vspace*{-2ex}
    \label{fig:csb_format}
\end{figure}

% \caleb{Caleb: Serif, see Edgar's comments. Perhaps we can use the citations in the section before the conclusion here?}
% \serif{If needed, we can replace it with the following:
\begin{comment}
\caleb{Caleb: I have replaced the text. Do we want to refer to Figure~\ref{fig:csb_format}?}\serif{if we are not referring to this, we can remove it?}
\end{comment}
Similar to~\cite{aktulga}, we parallelize $\B y=\B{A} \B{x}$ and $\B y=\B{A}^T \B x$ over the row-blocks and column-blocks, respectively.
Provided the tile size ($r\times k$) is selected carefully, the block of the input vector ($\B x$) processed by a tile and output vector $\B y$ corresponding to a row-block (updated by a single thread) is contained in private L1 or L2 caches. Figure~\ref{fig:csb_format} illustrates the parallelization of SpMV using CSB format.

\subsubsection{Implicit Representations}
\label{implicit_representation}

In each iteration of MWU, we perform one SpMV with the constraint matrix and with its transpose.
As seen in Section \ref{sec:MWUProbs}, the constraint matrices of many graph-based LPs are the vertex-edge incidence matrix of the graph.
We notice that the edge $(u,e)$ encoded in an incidence matrix can be described over over a vertex and a vertex pair $(u,(v,w))$ where the pair $(v,w)$ represents an edge in the adjacency matrix. In addition, in the incidence matrix, the value is $1$ if $u=v$ or $u=w$ and $0$ otherwise. Therefore, we can store the nonzero data in an vertex-edge incidence matrix $\B M$ implicitly as an adjacency matrix in memory. This reduces the memory cost of the constraint matrix by about half and therefore also reduce the number of accesses to the memory subsystem, particularly cache.

Having an implicit representation means that linear algebra operations on these matrices can also be expressed implicitly by formulating the computations on the adjacency matrix, $\B A$.

%Many matrices encoding graph constraints have special sparse structures and non-zero values. We can exploit these properties to save memory and cache capacity, as well as decrease the number of accesses to the memory subsystem. A simple application of implicit matrices is the SpMV of the interweaved identity matrix~\eqref{eq:interweaved} from densest subgraph~\eqref{eq:densest_linalg}. 

%For vertex-edge incidence~\eqref{eq:incidence} ($\B M$) and pair~\eqref{eq:oddshifted_incidence} ($\B O$) matrices, a $1$ exists at matrix location $(u,e)$ or $(u,2e+1)$, respectively, if and only vertex $u$ is incident to edge $e$. This structure permits us to

We emulate SpMVs for $\B M$ and $\B O$ by storing the edge, row, and column index for each non-zero in $\B A$, which we denote by $e$, $r$, and $c$, respectively. For example, we can compute $\B y=\B{M} \B x$ by accumulating element $x_e$ to $y_r$ and $y_c$ for each $(r,c,e)$ in $\B A$, and likewise for $\B y=\B{O}\B x$, then evaluating $y = y_r + y_c$. % in two steps: $y = y_r + y_c$, where $y_r[r] = x[e]$ and $y_c[c] = x[e]$ (and $y_r[r] = x[2e]$ and $y_c[c] = x[2e+1]$) $\forall (r,c) \in A$. 
To parallelize this SpMV while avoiding race conditions, we first traverse in row major order while reading the row indices of $\B A$, then in column major order while reading the column indices. % In the next subsection, we propose a suitable matrix format to conduct this traverse. 

We compute $\B y = \B{M}^T \B x$ by accumulating elements $x_r$ and $x_c$ to $y_e$ for all $(r,c,e)$ in $\B A$, and likewise for $\B y = \B O^T \B x$. Parallelizing these SpMVs are straightforward, since we accumulate to any given element of $\B y$ two and one times, respectively.

\subsubsection{Loop Fusion and Vectorization Opportunities}
% In order to achieve the maximum performance, we target acceleration of both of these primitives. For vector operations, we improve memory behavior by exploiting loop fusion and vectorization.
% \serif{referencing line numbers is not working?} \caleb{will try to get this working soon}
% The MWU algorithm, shown in Algorithm~\ref{alg:exp2}, has many vector operations with costly computations (e.g., calculating the exponential function) for each element. The gradient calculations (Lines~\ref{line:smax} and~\ref{line:smin}), finding a new step direction 
% \caleb{Line~\ref{line:get_maxd} is not a part of calculating the step direction. If we want to keep, should be its own item} 
% (Lines~\ref{line:direction}), and finding the step size (calculating $\alpha$ in Line~\ref{line:find_step_size}) are  examples. Furthermore, 
In each iteration of the MWU algorithm, we do several vector operations and our augmentation to step size search adds many vector operations, too. For some problems, such as vertex-cover or densest subgraph, the vectors in these operations have size $|E|$, which means they can be as costly as a single SpMV.

%Therefore, vector operations used in each iteration in, for example, gradient calculations (Lines~\ref{line:smax} and~\ref{line:smin}) and bang-for-buck value~\eqref{eq:b4b_conv}, can be as costly as a single SpMV.
% \serif{a better rewording here may shrink this. ideas?}
% \st{matrix-vector multiplications}. 
% Therefore, we apply loop fusion and vectorization optimizations for these calculations to improve their performance.
% gradient and new direction calculations, and for all other vector operations we created vectorized versions to maximize floating-point throughput of the given system.

% that can be implemented using SIMD instructions. Furthermore, we identified 5 regions of interests where we need temporary intermediate arrays to execute vector operations in parallel. 

% \begin{verbatim}
% *Compiler is not able to do it with multithreaded execution. 
%   cant detect the fusion opportunities
% *Significant reduction in memory operations
% *problem when incidence matrix is used
% \end{verbatim}

% $\nabla \mathrm{smax}_\eta(\B x)$

% \serif{will revisit this after, algo description in Sect4 completed}
Since these vector operations loop to apply simple arithmetic to each element in the vector, combining multiple vector operations in one pass via loop fusions can accelerate these methods. We identify two operations for fusion: (1) the gradient calculations using \emph{smax} and \emph{smin} (Lines~\ref{line:smax} and~\ref{line:smin}) 
% $\nabla \mathrm{smax}_\eta(\B{y})$ and $\nabla \mathrm{smin}_\eta(\B{y})$, 
and (2) the
calculation of the new step direction ~\eqref{eq:b4b_conv}. In both cases, loop fusion can reduce memory accesses and facilitate automatic vectorization.

\subsection{Distributed Parallelization}
%\serif{For a distributed environment, we consider SpMV implementations of implicit vertex-edge incidence ($\B{M}$), and implicit vertex-edge pair ($\B{O}$) matrices. These implicit operations can have reduced communication compared to an SpMV implementation if they are formed explicitly.?}

Distributed-memory parallelization of vector operations and explicit sparse matrix vector products in MWU can be done with standard techniques.
%However, implicit representations of vertex-edge incidence matrices is an important consideration for distributed-memory layout of the matrices.
Therefore, in this section, we will focus on describing and analyzing the benefits of the implicit representation for distributed-memory communication.

We use the same implicit representation described in Section \ref{implicit_representation}. We leverage a 2D matrix distribution of the adjacency matrix to perform implicit SpMVs with the incidence matrix. A 2D data layout is communication-efficient for matrix vector products since each processor computes on only $n/\sqrt{p}$ entries and contributes to $n/\sqrt{p}$ outputs (for an $n\times n$ matrix on a $\sqrt{p}\times \sqrt{p}$ processor grid). They are commonly employed for parallel processing of adjacency matrices~\cite{bulucc2011combinatorial}.

%With a 2D layout of the adjacency matrix, we perform vertex-edge incidence products with the same communication cost as a product with an adjacency matrix. 

With a 2D layout of the adjacency matrix, we perform vertex-edge incidence products with twice the communication cost of an adjacency matrix product. To do so, we store vector information corresponding to edges in the same processor layout as the adjacency matrix $A$. This means that for an edge $(u,v)$ in $A$, the machine owning the edge would store vector information for indices corresponding to $u$ and to $v$.

For simplicity, we assume a square processor grid.
With this approach, the product with the vertex-edge incidence matrix, $\B y = \B M \B x$, requires only a reduction of contributions to $\B y$ along rows and columns of the processor grid.
While for the product $\B y = \B M^T \B x$, only a broadcast of entries of $\B x$ along rows and columns of the processor grid is needed. In both cases, each processor sends or receives a subvector of size $O(\sqrt{n}/{p})$.

%\mina{We also notice that with multiple threads and processes, this form of implicit matrix-vector product changes the summation order and incurs floating point differences. This error can be fatal during the binary or Newton search step, i.e. if values of $\Phi$ or $\Psi$ become different between processors, so to tamp down on the error, we average these key values between processes.}

\section{Experimental Setup} \label{sec:exp_setup}
\subsection{System Setup}
We use Intel Knights Landing (KNL) nodes on the Stampede2 supercomputer as our testbed. 
Each KNL node has 68 1.4~GHz cores. Each core has a 32~KB private L1 cache, and 2 neighboring cores share a 1~MB L2 cache. KNL processors also support AVX2 and AVX-512 vector instructions. Each Stampede2 node has 112~GB of memory capacity with 96~GB DRAM and 16~GB MCDRAM used in cache mode.
% \serif{can we verify the it is in the cache mode?} \caleb{yea, is there a command I run to check?}\serif{i think if do numactl and see that there is a single node, it should be cache mode. or are there any webpage that details stampede architecture?}. \mina{its cache mode. flat mode has 2 nodes, and hybrid node isnt available. i also have numactl -H and it was oneThanks!}
% The summary of our system setup can be found in Table~\ref{tab:system}.
% \serif{will add the table}

% \begin{table}[h]
% \caption{System Setup}
% \label{tab:system}
% \footnotesize
% \centering
% \begin{tabular}{|l||p{0.3\textwidth}|}
% \hline
%  \textbf{Component}        & \textbf{Properties} \\ \hline \hline
% \textbf{Processor} & 68 cores @1.4~GHz, 32~KB L1 caches, 1~MB L2 caches shared by 2 cores, AVX-512 support  \\ \hline
% \textbf{Memory}    & 96~GB DRAM, 16~GB MCDRAM in cache mode \\ \hline
% \textbf{Software}  & OS:?  Kernel:? \\ \hline
% \end{tabular}
% \end{table}

\subsection{Implementations}
\textbf{MWU Implementations.} 

We implement two different versions of MWU ~\ref{alg:exp1}: (1) \mwupetsc, and
(2) \mwuomp. \mwupetsc relies on an efficient parallel BLAS library
PETSc~\cite{balay2019petsc} while \mwuomp is our hand-optimized implementation
using optimizations discussed in Section~\ref{sec:optimizations}. In our
implementation of Algorithm~\ref{alg:exp1}, we do not drop satisfied constraints
(i.e., we skip Line~\ref{line:term_step_size}). This simplifies the
implementation, and we did not find this affects convergence. We set $\epsilon =
0.1$ and terminated the algorithm if it exceeds 5000 iterations. To verify
correctness, we compare the solution from MWU with an exact solution, if
available. % MWU has two parameters, which we set to $\epsilon=0.1$ and
\texttt{max\_iter}=5000. These parameters control the accuracy and maximum
number of iterations of MWU, respectively, which were found by hand-tuning the
algorithm. To verify correctness, we compare the solution from MWU with an exact
solution, if available.

% See: https://petsc.org/main/include/petscblaslapack.h.html
% PETSc~\cite{balay2019petsc} is an efficient distributed sparse BLAS library with a Python interface (\texttt{petsc4py}).  
PETSc is a suite of data structures and routines for large-scale distributed operations, including vector and sparse matrix operations (which calls (sparse) BLAS under the hood), with a Python interface (\texttt{petsc4py})~\cite{balay2019petsc}. 
On Stampede2, we use PETSc with MKL version 19.1.1.  We use C++ for our \mwuomp
optimized implementation and OpenMP for parallelism.  We compile our code with
Intel compiler version 19.1.1 and enable \emph{-O3}, and \emph{-mAVX512}
compiler flags. We run the \mwuomp implementation by binding threads to physical
cores using \emph{numactl --physcpubind}. For \mwupetsc, we find that creating
$N$ processes each with 1 thread gives the best performance.

\textbf{General LP Solvers.} We compare our optimized MWU implementation to general LP solvers. We use \textit{IBM} \cplex~\cite{cplex2009v12} and
\gurobi~\cite{gurobi}, both in multi-threaded settings. If the problem is an ILP, we do not round the fractional solution.

% CPLEX exploits parallelism by running ... On the other hand, Gurobi runs 3 different methods: (1) AAA, (2) BBB, and (3) CCC. AAA and BBB are sequential methods while CCC can be executed in parallel. While running Gurobi, we use 1 thread for methods AAA and BBB and assign rest of threads in the system to method CCC.

% For CPLEX and Gurobi, we set the number of threads to 64.
For \cplex,
we set the run mode to \emph{opportunistic} to achieve the fastest
(but non-deterministic) run time. For \gurobi, we use the concurrent optimization 
setting, which concurrently runs \emph{primal simplex}, \emph{dual simplex}, and the
\emph{barrier method}. We report the fastest run time out of these three
methods. When the barrier method finishes first, we
report its run time \emph{before crossover} (unless otherwise noted) for a more fair comparison to MWU, which outputs fractional solutions. 
We implement all applications discussed in Section~\ref{sec:MWUProbs} with both \cplex and Gurobi.
Finally, we limit the solve time for both \cplex and \gurobi to 4 hours. All other parameters were set to defaults.

\textbf{Specialized Algorithms.} 
% Finally, we also run specialized algorithms 
We consider optimized custom implementations as baselines
for the two implicitly integral LP problems: bipartite matching
and densest subgraph. For the former, we use \graft~\cite{azad2016computing},
which employs the serial Karp-Sipser greedy initialization step~\cite{karp1981maximum} followed by a 
specialized breadth-first searches to find augmenting paths.
% process
% \textcolor{red}{Edgar: what is a grafting process?}
% to increase the bipartite matching until its reaches the maximum
% size. 
For the latter problem,
we used the Graph Based Benchmark Suite's~\cite{dhulipala2020graph} (\gbbs) 
approximate densest subgraph algorithm, which implements Charikar's greedy 2-approximation algorithm~\cite{charikar2000greedy}. Both \graft and \gbbs are implemented in C++ with OpenMP. We compile
\graft using OpenMP and the Intel compiler (version 19.1.1) with the \emph{-O2} flag. We compiled \gbbs with the g++ compiler version 9.1.0.

%\mina{talked w Caleb and skipped this}. There exist greedy algorithms with good approximation factors for vertex cover and dominating set ~\cite{young2001sequential}, however, we chose to not do a comparison against those because not only are greedy algorithms very fast, but there is no strong reason to prefer MWU or any general LP solver over them as vertex cover and dominating set have an integrality gap, so after rounding the solution to fractional LP, we would have this gap, instead of $(1+\epsilon)$, as our approximation factor. 

% \serif{need to say why is this different?}

% \begin{table}[]
% \caption{Applications and Implementations}
% \label{tab:apps_impls}
% \centering
% \footnotesize
% \begin{tabular}{|p{0.19\textwidth}||l|l|l|}
% \hline
%                                  & \textbf{CPLEX} & \textbf{Gurobi} & \textbf{Specialized} \\ \hline \hline
% \textbf{Dominating Set (DS)}    &                &                 &                      \\ \hline
% \textbf{Matching (M)}          &                &                 &                      \\ \hline
% \textbf{Bipartite Matching (BM)}&                &                 &                      \\ \hline
% \textbf{Vertex Cover (VC)}  &                &                 &                      \\ \hline
% \textbf{Densest Subgraph (DSub)}&                &                 &                      \\ \hline
% \textbf{Optimal Transport (OT)} &                &                 &                      \\ \hline
% \end{tabular}
% \end{table}

% \paragraph{Input Datasets}
\subsection{Input Graphs \label{sec:inputgraph}}
We select a variety of real-world and synthetic undirected graphs from the SuiteSparse Matrix Collection~\cite{suiteSparse} and list
them in Table~\ref{tab:GraphDescription}. 
% \serif{I don't we need to say this because edge weights are defined by the problem}We set all edge weights to one. 
% We write the five real-world graphs in Table~\ref{tab:GraphDescription}. 
Our real-world graphs come from diverse domains, such as a road, social, and user-product network.
We also use two sets of synthetic graphs, the 
first set being random geometric graphs (\mtxname{rgg}) which have a planar-like structure, and the second set being Kronecker graphs (\mtxname{kron}) from Graph500 which show a strong community structure. 
% We also consider two sets of synthetic graphs from the Graph500
% benchmark. We include 10 random geometric graphs, labeled
% by rgg-Y as well as 6 Kronecker graphs, labeled by 
% kron-X. 

% To assess the effectiveness of our algorithmic improvements from Section~\ref{sec:AlgImprove}, we use the modestly-sized \mtxname{rgg-18} due to run time limitations.
% In scalability experiments for MWU, we use 4 large graphs: \mtxname{com-Orkut}, \mtxname{hollywood-2009}, as well as the largest \mtxname{kron}, and \mtxname{rgg} graphs. 
% \st{On the other hand} 
% {When comparing MWU to other optimization libraries}, we use all input graphs. 
Note that none of the graphs we selected are bipartite, which is required in \bimatching. To obtain bipartite graphs, we read the input adjacency matrix as a biadjacency matrix, meaning that the rows and columns of the matrix correspond to the left and right
% \textcolor{red}{Edgar:  'partite' is not a word, in general use different words than the parts of the name of the term itself to explain a given term.}
sets of vertices, respectively, where edges can only go between between vertices in different sets.
% \st{when comparing MWU against state-of-the-art LP solvers and application specific algorithms.}
% \caleb{One application from Section~\ref{sec:MWUProbs} we test is bipartite matching}, but
% none of our selected graphs are bipartite. 
% To obtain bipartite graphs for \bimatching problem,}
% \caleb{Our solution is to} read the input adjacency matrix as a biadjacency matrix, 
% meaning that the rows and columns of the matrix correspond to the left and right partite, respectively. 
% \serif{after reading section VII-D: we need to add a note about how we obtain bipartite graphs from these graphs to this section}

\begin{table}[h]
\footnotesize
\caption{List of real-world and synthetic graphs}
\centering
\begin{tabular}{|l||r|r|} \hline
 \textbf{Graphs (Abv.)} & \textbf{$|V|$} & \textbf{$|E|$} \\ \hline \hline
 \mtxname{usroads (usroads)} & 129,164 & 330,870  \\ \hline
 \mtxname{com-Amazon (amazon)} & 334,863 & 1,851,744 \\   \hline
 \mtxname{coPapersCiteseer (papers)} & 434,102 & 32,073,440 \\ \hline
 \mtxname{hollywood-2009 (hollyw)} & 1,139,905 & 113,891,327 \\ \hline
 \mtxname{com-Orkut (orkut)} & 3,072,441 & 234,370,166 \\ \hline 
 \mtxname{kron-X} & X=$2^{17}$-$2^{21}$ & $\approx$X$\times 80$ \\ \hline
 \mtxname{rgg-Y} & Y=$2^{17}$-$2^{24}$ & $\approx$Y$\times 15$ \\ \hline
%  \bottomrule
\end{tabular}

\label{tab:GraphDescription}
\end{table}

\section{Experimental Results}
\label{Section:experimental_results}

In this section, first, we compare our implementation to state-of-the-art software including general LP solvers, \cplex and \gurobi, to specialized parallel implementations for particular graph problems, \mina{ and to the parallel implementation of another multiplicative weights update algorithm from Makari et. al ~\cite{makari2013distributed}.}

Then, we evaluate the effectiveness of our algorithmic improvements and software optimizations.
% and analyze how it impacts the performance of MWU. 
We start by finding how the 
incorporation of a step size search reduces the number of MWU iterations. We then test the performance improvements
from our software optimizations and its scalability by comparing performance between \mwupetsc
and \mwuomp.
% compared to a parallel MWU implementation with an efficient BLAS library (PETSc).
% implementation to achieve parallelism. 
%removed
% This section
% culminates in a detailed comparison of our fastest MWU algorithm\serif{implementation} compared to other 
% state-of-the-art software
% % , which will include both 
% including general LP solvers and 
% specialized codes targeted at particular graph problems. 
%removed

% \subsection{Test Graphs}

%--------------------------------------------------------------------

\subsection{Comparison of MWU to Other Algorithms}
% Table~\ref{tab:comparison_all}
\begin{table*}[htb]
\scriptsize
% \footnotesize
\caption{Run time (in seconds) of \mwuomp compared to other optimization libraries and custom applications. Cells with a dash indicate the algorithm
took 4+ hours to run or had a memory error, with the exception
of \vcover with \gurobi on all the \mtxname{kron} graphs and \mtxname{hollyw.}, which 
had a \texttt{ConstraintError}. }
\label{tab:comparison_all}
\centering
\begin{tabular}{|r||l|l|l|l||l|l|l||l|l|l||l|l|l|l|}
\hline
% \multicolumn{1}{|l||}{} & \multicolumn{3}{c||}{\textbf{Matching}}                                                                         & \multicolumn{4}{c||}{\textbf{Bipartite Matching}}                                                                                                     & \multicolumn{3}{c||}{\textbf{Dominating Set}}                                                                   & \multicolumn{6}{c||}{\textbf{Densest}}                                                                                                                                                                                                               \\ \hline
\multicolumn{1}{|l||}{} &  \multicolumn{4}{c||}{\bimatching}  
& \multicolumn{3}{c||}{\domset} 
& \multicolumn{3}{c||}{\vcover} 
& \multicolumn{4}{c|}{\densest}  \\ \hline 
\multicolumn{1}{|l||}{} & 
% \multicolumn{1}{c|}{\textbf{OMP}} & \multicolumn{1}{c|}{\textbf{CPLEX}} & \multicolumn{1}{c||}{\textbf{Gurobi}} & 
\multicolumn{1}{c|}{\textbf{MWU}} & \multicolumn{1}{c|}{\textbf{CPLEX}} & \multicolumn{1}{c|}{\textbf{Gurobi}} & \multicolumn{1}{c||}{\textbf{graft}} & 
\multicolumn{1}{c|}{\textbf{MWU}} & \multicolumn{1}{c|}{\textbf{CPLEX}} & \multicolumn{1}{c||}{\textbf{Gurobi}} & 
\multicolumn{1}{c|}{\textbf{MWU}} & \multicolumn{1}{c|}{\textbf{CPLEX}} & \multicolumn{1}{c||}{\textbf{Gurobi}}& 
\multicolumn{1}{c|}{\textbf{MWU}} & \multicolumn{1}{c|}{\textbf{CPLEX}} & \multicolumn{1}{c|}{\textbf{Gurobi}} &  \multicolumn{1}{c|}{\textbf{gbbs}} \\ \hline \hline
\mtxname{rgg-15} & 0.08 & 9.15 & 6.93 & 0.07 & 0.38 & 5.56 & 3.60 & 2.38 & 3.99 & 7.88 & 0.16 & 24.01 & 12.90 & 0.04 \\ \hline
\mtxname{rgg-16} & 0.09 & 21.23 & 14.38 & 0.14 & 0.34 & 15.62 & 10.64 & 3.37 & 10.72 & 6.86 & 0.22 & 54.57 & 25.61 & 0.05 \\ \hline
\mtxname{rgg-17} & 0.24 & 45.35 & 31.68 & 0.32 & 0.68 & 43.11 & 22.91 & 5.65 & 28.98 & 15.90 & 0.54 & 141.95 & 119.43 & 0.09 \\ \hline
\mtxname{rgg-18} & 0.15 & 114.31 & 76.76 & 0.71 & 3.54 & 561.23 & 49.01 & 9.85 & 82.43 & 38.42 & 0.66 & 349.34 & 344.72 & 0.13 \\ \hline
\mtxname{rgg-19} & 0.54 & 283.79 & 170.81 & 1.86 & 2.08 & 3,045.28 & 111.90 & 23.48 & 114.63 & 86.94 & 1.41 & 1,202.33 & 877.80 & 0.21 \\ \hline
\mtxname{rgg-20} & 0.43 & 716.66 & 406.77 & 4.44 & 7.79 & - & 255.15 & 44.84 & 292.40 & 226.01 & 2.34 & 4,081.44 & 2,017.68 & 0.32 \\ \hline
\mtxname{rgg-21} & 0.85 & 2,186.03 & 917.63 & 11.12 & 17.25 & - & 555.60 & 87.95 & 659.18 & 504.65 & 5.15 & - & - & 0.56 \\ \hline
\mtxname{rgg-22} & 2.81 & - & - & 30.67 & 46.37 & - & 1,313.81 & 183.67 & - & - & 13.21 & - & - & 0.90 \\ \hline
\mtxname{rgg-23} & 3.92 & - & - & 80.99 & 247.68 & - & - & 367.86 & - & - & 22.40 & - & - & 1.69 \\ \hline
\mtxname{rgg-24} & 21.00 & - & - & 226.12 & 115.02 & - & - & 856.06 & - & - & 74.87 & - & - & 3.18 \\ \hline
\mtxname{kron-16} & 3.80 & 95.61 & 136.6 & 0.00 & 1.38 & 19.22 & 23.22 & 4.53 & 81.83 & - & 1.81 & 3,169.42 & - & 0.11 \\ \hline
\mtxname{kron-17} & 1.87 & 200.79 & 335.18 & 0.01 & 10.97 & 63.81 & 59.44 & 55.16 & 194.90 & - & 2.85 & 8,053.81 & - & 0.17 \\ \hline
\mtxname{kron-18} & 2.60 & 462.66 & 642.27 & 0.01 & 3.07 & 214.97 & 155.06 & 33.63 & 414.28 & - & 10.53 & - & - & 0.26 \\ \hline
\mtxname{kron-19} & 4.38 & - & - & 0.01 & 10.24 & 657.11 & 379.04 & 53.06 & 2,354.63 & - & 12.88 & - & - & 0.43 \\ \hline
\mtxname{kron-20} & 10.36 & - & - & 0.02 & 32.51 & 2,292.96 & 1,048.8 & 82.93 & 3,091.02 & - & 24.30 & - & - & 0.67 \\ \hline
\mtxname{kron-21} & 32.33 & - & - & 0.04 & 584.85 & 7,072.73 & 2,342.73 & 210.92 & - & - & 56.28 & - & - & 1.13 \\ \hline
\mtxname{usroads} & 0.07 & 20.24 & 12.30 & 0.04 & 1.49 & 16.30 & 8.85 & 1.07 & 12.31 & 5.22 & 0.27 & 35.25 & 16.80 & 0.03 \\ \hline
\mtxname{amazon} & 0.71 & 93.93 & 126.34 & 0.05 & 22.14 & 115.10 & 98.99 & 2.13 & 72.38 & - & 1.52 & 750.42 & 1,040.63 & 0.09 \\ \hline
papers & 6.44 & 3,549.09 & 542.19 & 0.33 & 10.14 & 35.14 & 49.15 & 149.95 & 767.30 & 443.05 & 3.78 & - & 3,945.37 & 0.39 \\ \hline
\mtxname{hollyw.} & 39.13 & - & - & 0.56 & 43.86 & - & - & 130.50 & - & -& 24.79 & - & - & 1.89 \\ \hline
\mtxname{orkut} & 29.99 & - & - & 29.12 & 162.24 & - & - & 334.70 & - & - & 201.09 & - & - & 3.18 \\ \hline
\end{tabular}
\end{table*}

% Table~\ref{tab:ot_comparison_all}
 
We now compare the \mwuomp implementation of MWU with Newton's method and all the software implementation optimizations to other state-of-the-art optimization libraries and custom implementations (\graft for \bimatching and \gbbs for \densest). All experiments are run with 64 threads on a single KNL node. Table~\ref{tab:comparison_all} shows the execution times to find $(1+\epsilon)$-relative solutions for four positive LPs on various graphs \mina{where $\epsilon=0.1$}. A ``-'' in a cell means that the input graph was either too large to be processed by the library, or the run time exceeded 4 hours. 

\noindent\textbf{Comparison with Exact LP Solvers.}

\mina{For all LP solvers, we do not do rounding as post-processing. Therefore, for the exact solvers, we have integral solutions to graphs problems that have implicitly integral LPs, and exact fractional solutions for the relaxed LPs with integrality gaps. For approximate solvers, we are not guaranteed an integral solution on integral LPs. Note that \cplex and \gurobi return exact solutions for the target LPs, while MWU finds an $(1+\epsilon)$-relative solution with a target value of $\epsilon=0.1$.} We find that our \mwuomp implementation is able to find an $\epsilon=0.1$ solution in all cases except \bimatching problem with \mtxname{rgg-20}. However, even for this case, our error rate is 0.104.

Our results show \mwuomp consistently outperforms \cplex and \gurobi libraries. 
For \bimatching, \domset, \vcover, and \densest graph LPs, \mwuomp is up to 2548x (\mtxname{rgg-21}), 1482x (\mtxname{rgg-19}), 43x (\mtxname{kron-19}), and 2860x (\mtxname{kron-17}) faster than \cplex, respectively. Although \gurobi is generally faster than \cplex, \mwuomp outperforms \gurobi by up to 1070x (\mtxname{rgg-21}), 55x (\mtxname{rgg-19}), 5x (\mtxname{rgg-21}), and 816x (\mtxname{rgg-20}) for \bimatching, \domset, \vcover, and \densest graph LPs, respectively, before crossover occurs. When
comparing the time when after crossover or one of the simplex methods from \gurobi terminates (whichever comes first), the relative speedups are
1462x (\mtxname{rgg-21}), 3510x (\mtxname{rgg-19}), 10x (\mtxname{usroads}), and 878x (\mtxname{rgg-20}) for \bimatching, \domset, \vcover, and \densest graph LPs, respectively.
In addition to significant speedups, we also observe that MWU is capable of running much larger problems. For instance, both \cplex and \gurobi can only solve \mtxname{kron-21} for \domset but not for \densest, the latter which contains three times
more nonzeros than the former. Moreover, both
LP solvers fail to solve any of the problems with the largest graphs, \mtxname{hollywood} and \mtxname{orkut}. 
% \serif{\cplex and \gurobi fails for different graphs with different applications? because of LP size? how does it affect the solve time?}.

% \vspace{40}

%-----------------------------------------
%-----------------------------------------

\noindent\textbf{Comparison with Custom Implementations.}
%The MS-BFS algorithm for bipartite matching and Charikar's greedy algorithm for densest subgraph implemented in \gbbs return integral solutions. \mwuomp returns possibly fractional approximations for both \bimatching and the specialized formulation for \densest described in Section \ref{sec:MWUProbs}, which are implicitly integral matrices. 
\mina{ We compare \mwuomp performance to \graft for \bimatching and to \gbbs for
\densest. The MS-BFS algorithm returns an exact solution. \mwuomp returns an
approximate solution with $\epsilon=0.1$.} The MS-BFS algorithm
~\cite{azad2016computing} initializes with a serial Karp-Singer greedy step and
finds augmenting paths in parallel using specialized BFS. The performance of
\graft heavily depends on the graph structure.
In general, we observed the \mwuomp can outperform \graft for graphs with planar
structures by 1.8-22.4x for the \mtxname{rgg} graphs and \mtxname{usroads}. On
the other hand, for graphs which contain a strong community-structure
or vertices with high degrees, \graft outperforms MWU. For example, amongst the
\mtxname{kron} graphs, \graft can be up to 450x faster than \mwuomp. For these
types of graph instances, \bimatching generally spends much less time on the
grafting process than with planar-structured graphs, which we often find to be
the dominating cost of \bimatching. 

For \densest, \gbbs implements Charikar's greedy 2-approximation algorithm for densest subgraph~\cite{charikar2000greedy}, but the relative error is usually much better in practice (but worse than $\epsilon=0.1$). \mina {Again, we run \mwuomp with $\epsilon=0.1$.} We observe that \gbbs always outperforms MWU, achieving a maximum speedup of 63.2x and minimum speedup of 4x.

%  gbbs's implementation of Charikar's greedy 2-approximation algorithm. On the other hand, we find that gbbs
% always outperforms MWU.

% For \bimatching, starting with \mtxname{rgg-22}, \mtxname{kron-19}, \mtxname{hollywood}, and \mtxname{com-Orkut} both \cplex and \gurobi fail to solve the LP, whereas MWU solves it in seconds, and always under a minute.
% a minute. Furthermore, we observe up to 3 orders of magnitude speedup compared \cplex and \gurobi with \mwuomp.

%\section{Other Related Work}
% \serif{needs to be revisited}
% In this work, we target 
% % loop fusion and 
% acceleration of Sparse-Matrix Vector Multiplications for solving graph problems with positive LPs. 

% -------------------------------------
% -------------------------------------
% -------------------------------------

\noindent\textbf{Comparison with Previous Work.}
\mina{ We compare MWU (Algorithm~\ref{alg:exp2}) and an implementation of a gradient descent algorithm with adaptive error~\cite{makari2013distributed}, which uses a less theoretically efficient multiplicative weights update algorithm but is the only other distributed study of multiplicative weights update methods on graph problems. Their paper compared their algorithm, called MPCSolver, against their implementation of Young's parallel algorithm for feasibility generalized matching, which is a mixed packing and covering LP, for an $(1+\epsilon)$-relative solution ($\epsilon=0.05$) and found that the implementation of MPCSolver outperformed Young's algorithm}. We provide a detailed description of the problem, datasets, and gradient descent algorithm in Appendix~\ref{sec:more_details_on_gm}.

% We refer the reader to~\cite{makari2013distributed} on how to formulate this problem. 
% We consider the number of MWU iterations to find an  using 256 MPI processes, each with 1 thread. 

% which have, after pre-processing, 473k and 1.6m vertices, as well as 100m and 252m edges, respectively. 

\mina{ We run the same experiment as them with \mwuomp, using the same datasets as well: the Netflix and KDD datasets~\cite{bennett2007netflix,dror2012yahoo}. Because both algorithms solve the same LP with a multiplicative weight update approach, there are only minor differences in vector operations between the two algorithms. Consequently, for this section, we compare iteration counts rather than time between the two algorithms.}
The two proposed algorithms are compared in Figure~\ref{fig:MPC_compare}. For MWU, we consider both the standard step size and the Newton's method
\mina{for step size search.} The gradient descent data is manually extracted from~\cite{makari2013distributed} using WebPlotDigiter~\cite{Rohatgi2020}. 
The plot shows both MWU with Newton's method and gradient descent with adaptive error find a $(1+\epsilon)$-relative solution in less than 2000 iterations, whereas MWU with standard step size converges much more slowly. 
% In fact, MWU with a standard step size has a max violation of $0.95$ and $0.90$ for the Netflix and KDD dataset, respectively, after about 43k and 28k iterations. 
Moreover, MWU with Newton's method incurs $10 \times$ and $41 \times$ fewer iterations than gradient descent for Netflix and KDD, respectively. 

These results highlight the effectiveness of using a step size strategy, such as Newton's method, over the standard step size. Furthermore, MWU with Newton's method converges more rapidly than gradient descent with an adaptive error. However, since both methods use heuristics to accelerate the method, testing additional positive LPs and datasets would be needed for a comprehensive understanding of the trade-offs between MWU and gradient descent. \mina{ The heuristic that MPCSolver uses prematurely stops the algorithm once it detects that the per-iteration decrease in constraint violation falls below a threshold.} 

% Once again, we see the effectiveness of line search. 

\begin{figure}[h]
    \centering
    \includegraphics[scale=0.5]{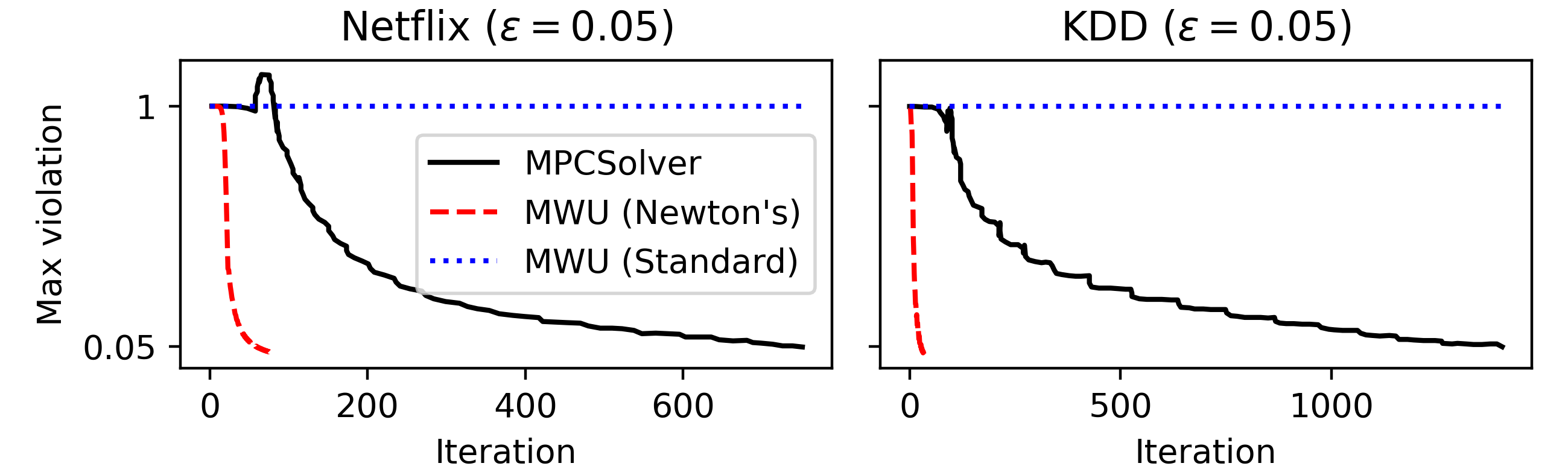}
    % \vpsace*{-1em}
    \caption{Max violation, defined as $\max\{0,\max(\B{Px})-1,1-\min(\B{Cx})\}$ for MPCSolver, which is a gradient descent algorithm with adaptive error~\cite{makari2013distributed}, and MWU (Algorithm~\ref{alg:exp2}) with standard step size and Newton's method.}
    \vspace*{-2ex}
    \label{fig:MPC_compare}
\end{figure}

%--------------------------------------------------------------------

\subsection{Parallel Scalability of MWU \label{sec:MWUScale} }

We now analyze the strong scaling behavior of \mwupetsc and \mwuomp.
%implementation on shared memory and the strong scaling of the \mwupetsc and \mwuhybrid extensions on distributed memory with our largest graphs, \mtxname{com-Orkut}, \mtxname{hollywood}, \mtxname{kron-21}, and \mtxname{rgg-24}.
%
Figure~\ref{fig:scalability} displays the speedup with respect to single-threaded execution of the \mwuomp implementation. 
%Figure~\ref{fig:scalability} is extended to 256 and 1024 threads with \mwuhybrid, with 16 OMP threads per MPI process and 64 total pinned threads on a machine, for two representative problems of 

When executing on a single node, all LP problem types are run along a range of thread counts from single threaded to 68 threads, which is the maximum number of hardware threads on one machine. Here, \mwuomp is able to achieve speedup over 16x with 68 threads in 90\% of experiments and over 32x in 50\% of experiments. Overall, the \mwuomp implementation achieves speedups of 13-55x on 68 threads compared to its single threaded run times. The largest differences between \mwuomp and \mwupetsc are observed on graph applications where we use specialized matrices such as \vcover, \bimatching, and \densest problems. High parallel efficiency in \mwuomp is achieved due to load balancing and high locality in matrix-vector multiplications of transposed specialized matrices and vector operations.
On the other hand, we see that the \mwuomp can only achieve 2-3x speedup compared to \mwupetsc for dominating set LP (\domset) application. Note that, for \domset, \mwuomp can only make use of format selection and memory access minimization optimizations for SpMV operations.

\mina{For multi-node results, we execute \mwuomp and \mwupetsc with 64 MPI processes per node and 1 thread per process (this was the most performant configuration for \mwupetsc on a single machine). We only run experiments where the total number of processes is square, as that is a requirement for our implicit representation.} We run all algorithms with Newton's optimization for step size search and limit Newton's methods to 5000 iterations. \mina{On distributed memory, except for vertex cover on \mtxname{rgg-24}, \mwuomp runs faster than \mwupetsc at scale.
For the distributed problem, we also observe almost linear scaling for all graphs except \mtxname{rgg-24} in \mwuhybrid. A matrix-vector product on the incidence matrix of banded matrices, like \mtxname{rgg-24}, reduces communication on a 1D data distribution pattern, so a 1D parallelization (e.g., row-wise distribution of the matrix) is more efficient than a 2D distribution, which is the layout we use.}

%In our 2D communication scheme, as more processors are added, it greatly reduces each processor's local work at the cost of increasing each processor's communication by a small amount. However, since \mtxname{rgg-24} is a sparse matrix with a narrow band, the size of local work is already dwarfed by 2D communication costs so increasing the processor count doesn't net much gain.

\mina{For \mwupetsc, we observe good scaling on \domset, which we expect as the LP for this problem does not use implicit representation and is a 1D problem. We observe poor scaling in all other problems on all graphs except \mtxname{rgg-24}. For \mtxname{rgg-24}, \mwupetsc achieves good performance due to its internal representation \cite{balay2019petsc}, which, we believe communicates only the vector entries needed by each processor based on the sparsity pattern of rows assigned to it. However, \mwupetsc performs extremely poorly on \densest when we use multiple processors and does not complete in under 2 hours.}
In conclusion, for general graphs, the implicit 2D representation scales well compared to a explicit 1D representation.

\begin{figure*}[h]
    \centering
    \captionsetup[subfloat]{farskip=0pt,captionskip=0pt}
    \subfloat[\domset \mtxname{orkut}]{\includegraphics[width=0.19\textwidth]{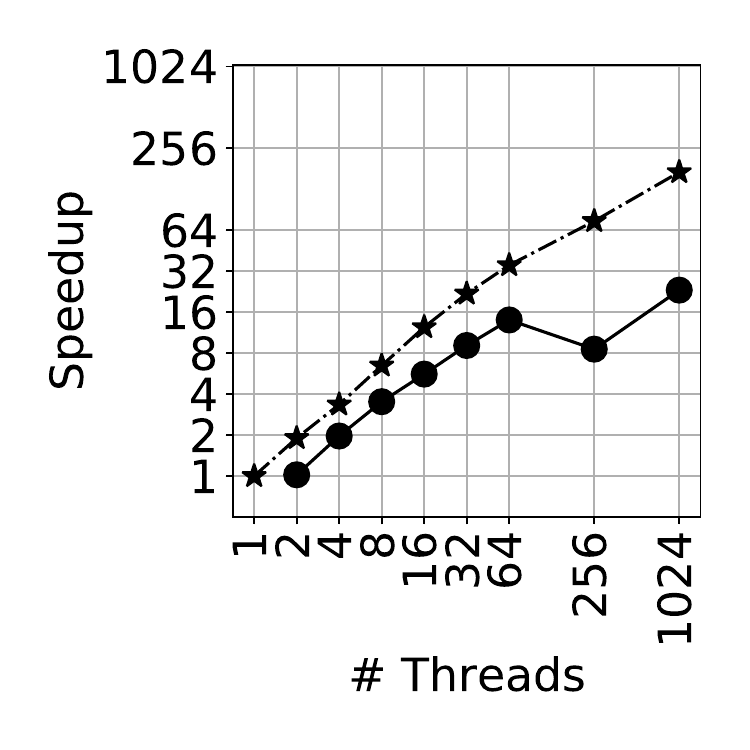}\label{fig:domset_orkut}}
    ~
    \subfloat[\domset \mtxname{hollyw.}]{\includegraphics[width=0.19\textwidth]{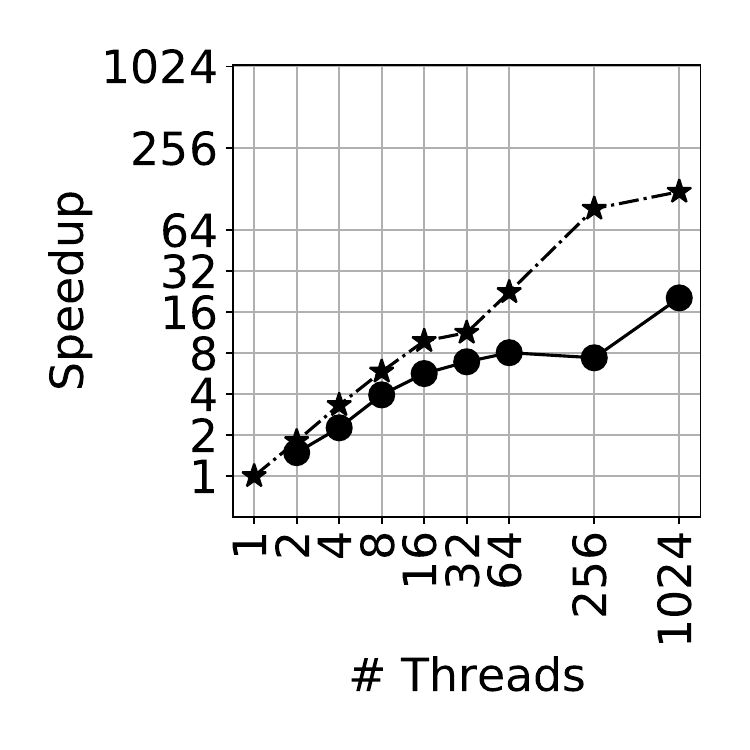}\label{fig:domset_hollywood}}
    % \subfloat[DomSet kron]{\includegraphics[width=0.2\textwidth]{figs/scal_petsc_omp/dom-set_kron-21.pdf}\label{fig:high_aft_mkl}}
    % ~
    % \subfloat[DomSet rgg]{\includegraphics[width=0.2\textwidth]{figs/scal_petsc_omp/dom-set_rgg-24.pdf}\label{fig:high_aft_mkl}}
    % \\
    % ~
    % \subfloat[Matching com-Orkut]{\includegraphics[width=0.2\textwidth]{figs/scal_petsc_omp/adj-matching_com-Orkut.pdf}\label{fig:low_aft_mkl}}
    ~
    \subfloat[\matching \mtxname{hollyw.}]{\includegraphics[width=0.19\textwidth]{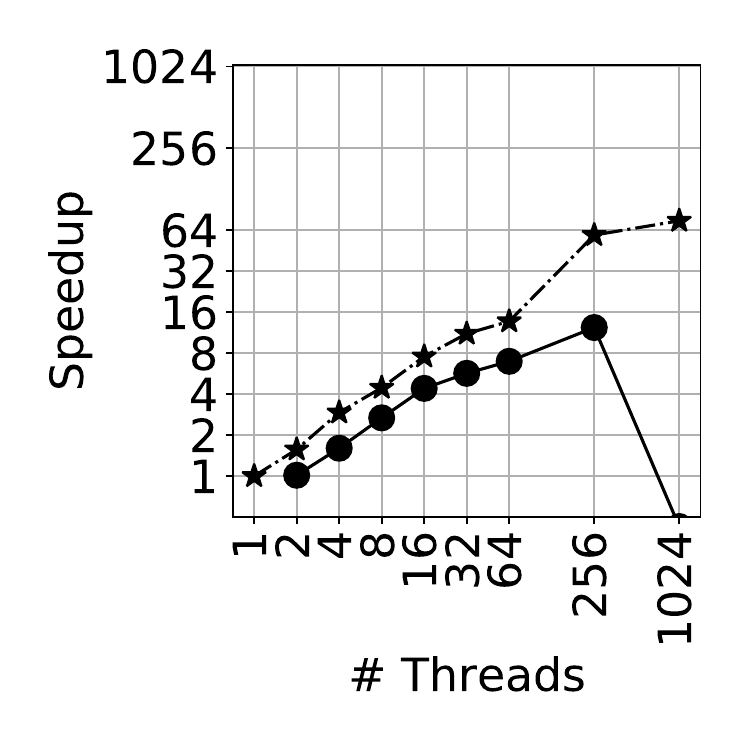}\label{fig:matching_hollywood}}
    ~
    \subfloat[\matching \mtxname{kron-21}]{\includegraphics[width=0.19\textwidth]{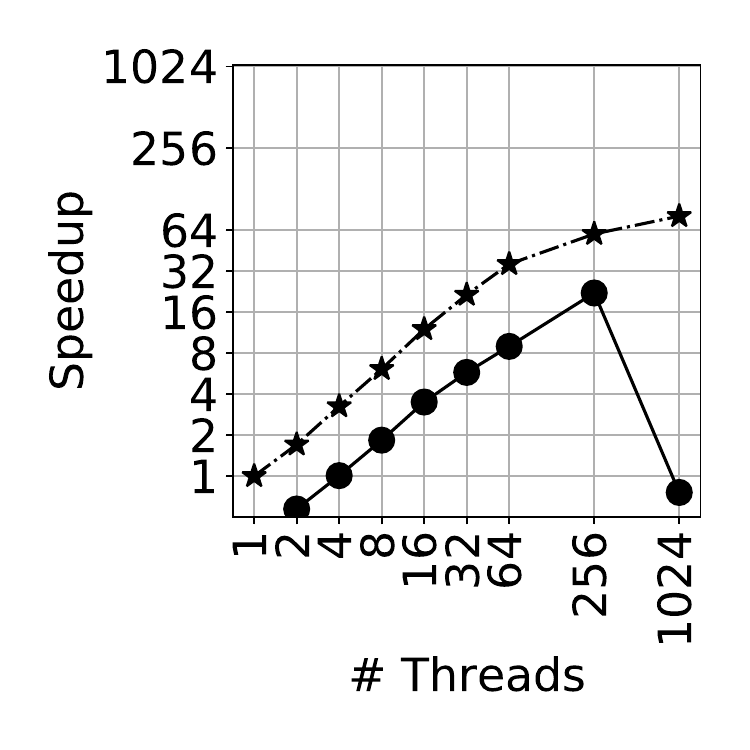}\label{fig:matching_kron}}
    % ~
    % \subfloat[Matching rgg]{\includegraphics[width=0.2\textwidth]{figs/scal_petsc_omp/adj-matching_rgg-24.pdf}\label{fig:high_aft_mkl}}
    % \\
    % \subfloat[BiMatching com-Orkut]{\includegraphics[width=0.2\textwidth]{figs/scal_petsc_omp/biadj-matching_com-Orkut.pdf}\label{fig:low_aft_mkl}}
    ~
    \subfloat[\bimatching \mtxname{hollyw.}]{\includegraphics[width=0.19\textwidth]{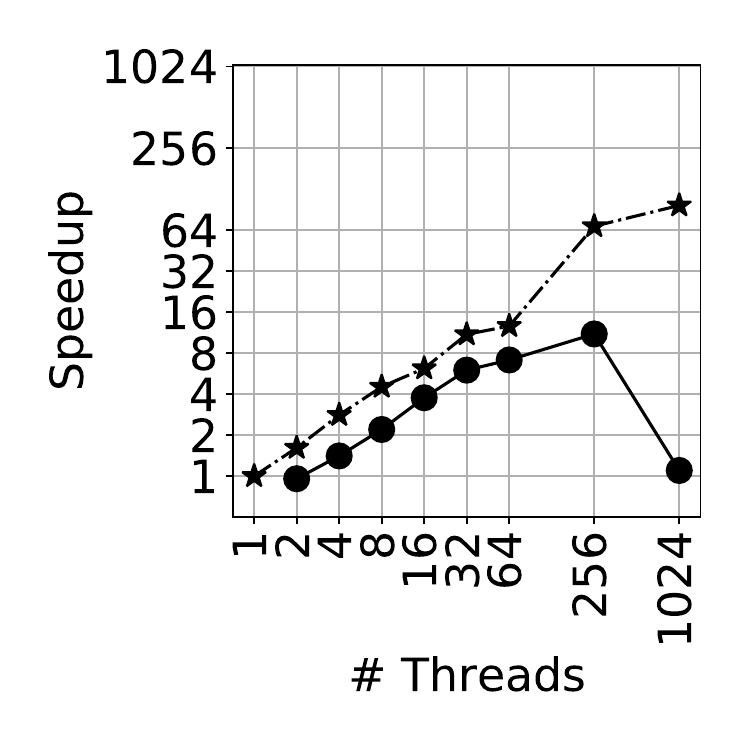}\label{fig:bimatching_hollywood}}
    \\
    \subfloat[\bimatching \mtxname{kron-21}]{\includegraphics[width=0.19\textwidth]{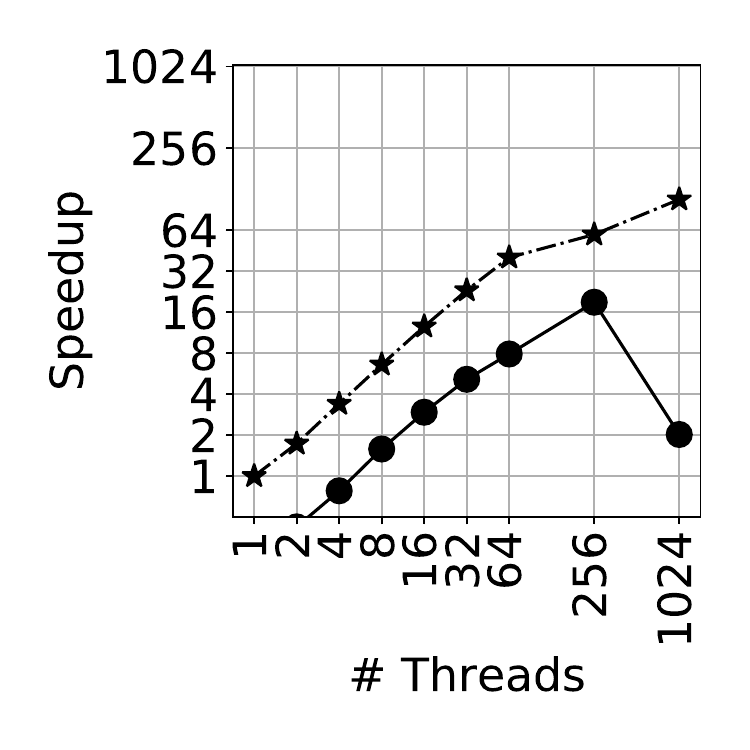}\label{fig:bimatching_kron}}
    ~
    % \subfloat[BiMatching rgg]{\includegraphics[width=0.2\textwidth]{figs/scal_petsc_omp/biadj-matching_rgg-24.pdf}\label{fig:high_aft_mkl}}
    % \\
    \subfloat[\vcover \mtxname{orkut}]{\includegraphics[width=0.19\textwidth]{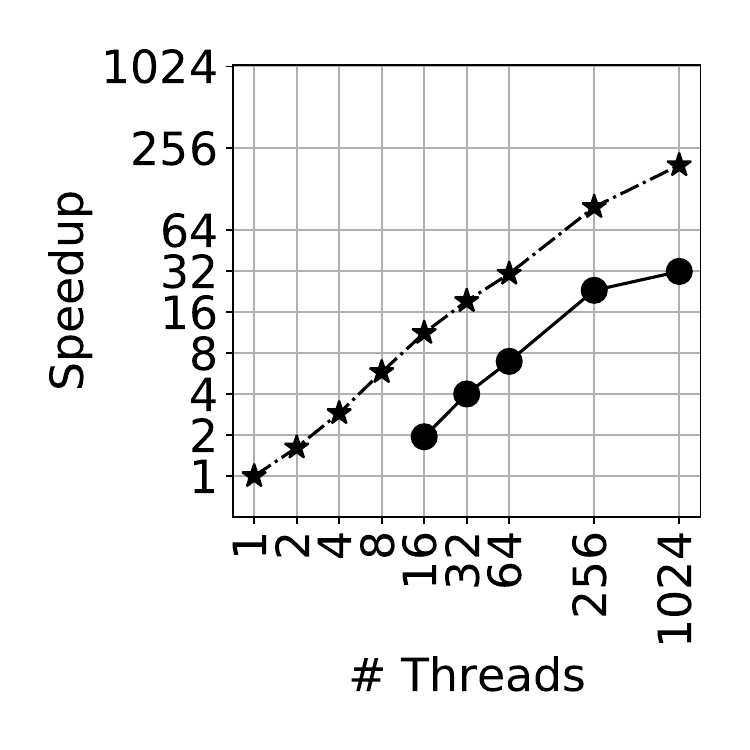}\label{fig:vtxcover_orkut}}
    ~
    % \subfloat[Vtx-Cover hollywood]{\includegraphics[width=0.2\textwidth]{figs/scal_petsc_omp/vtx-cover_hollywood.pdf}\label{fig:high_aft_mkl}}
    % ~
    % \subfloat[Vtx-Cover kron]{\includegraphics[width=0.2\textwidth]{figs/scal_petsc_omp/vtx-cover_kron-21.pdf}\label{fig:high_aft_mkl}}
    % ~
    \subfloat[\vcover \mtxname{rgg-24}]{\includegraphics[width=0.19\textwidth]{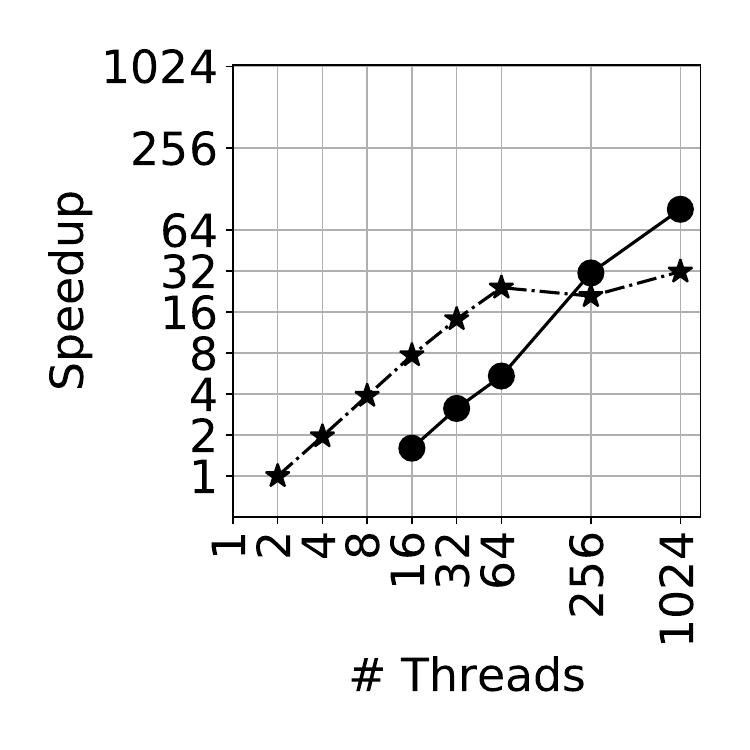}\label{fig:vtxcover_kron}}
    % \\
    % \subfloat[Densest com-Orkut]{\includegraphics[width=0.2\textwidth]{figs/scal_petsc_omp/densest-sub_com-Orkut.pdf}\label{fig:low_aft_mkl}}
    ~
     \subfloat[\densest \mtxname{orkut}]{\includegraphics[width=0.19\textwidth]{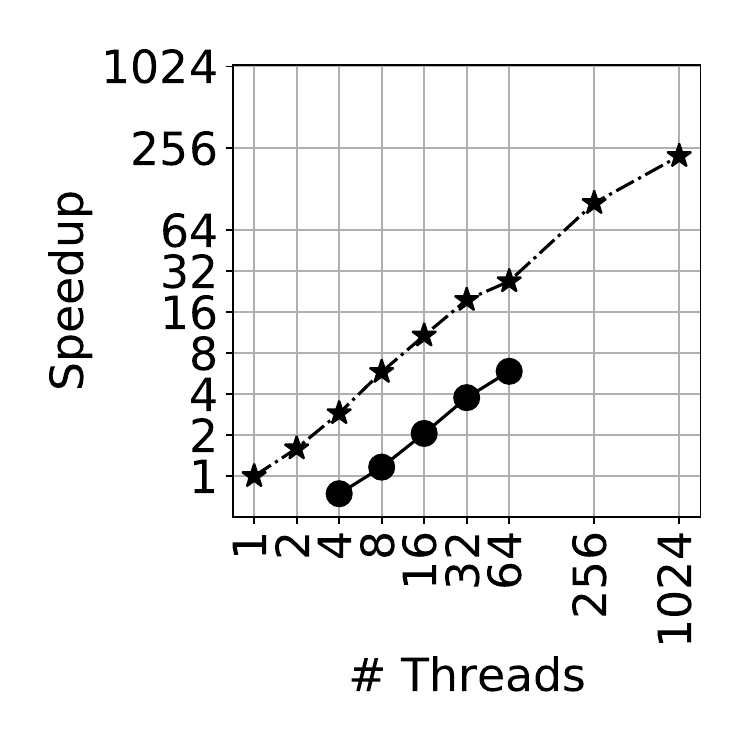}\label{fig:densest_orkut}}
    %\subfloat[\densest \mtxname{hollyw.}]{\includegraphics[width=0.19\textwidth]{figs/scal_petsc_omp_v4/densest-sub_hollywood.pdf}\label{fig:densest_hollywood}}
    % ~
    % \subfloat[Densest kron]{\includegraphics[width=0.2\textwidth]{figs/scal_petsc_omp/densest-sub_kron-21.pdf}\label{fig:high_aft_mkl}}
    ~
    \subfloat[\densest rgg-24]{\includegraphics[width=0.19\textwidth]{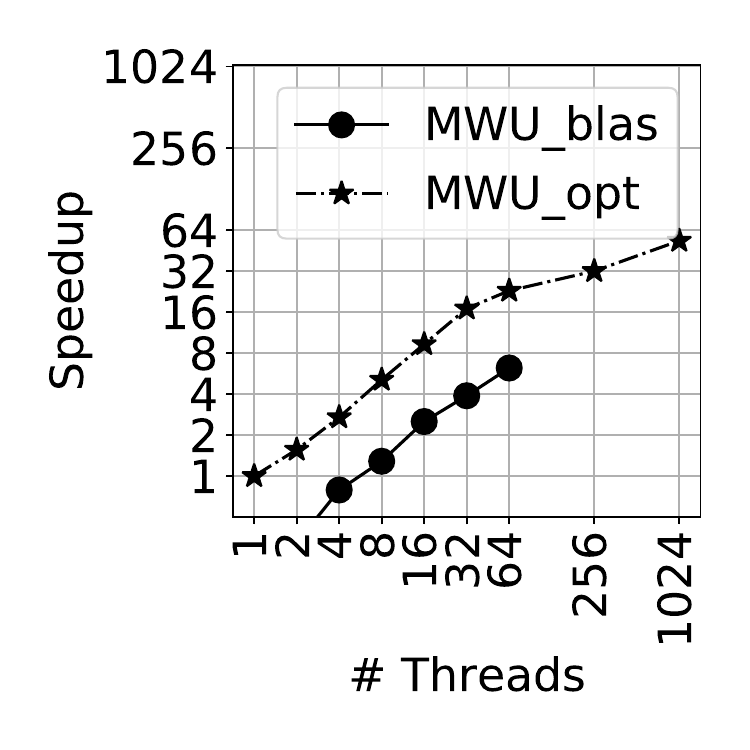}\label{fig:densest_rgg}}
    
    \caption{Scalability of \mwuomp and \mwupetsc (see subfigure (j) for the legend). Axes are in $log_2$ scale. All values are normalized to single-threaded execution of \mwuomp. We omit the results for \mwupetsc, if the execution time is slower than single-threaded execution of \mwuomp.}
    \label{fig:scalability}
\end{figure*}

\subsection{Improvements from Step Size Strategy} \label{sec:sec:step_strat_exp}
We first evaluate the effectiveness of step size search (Section~\ref{sec:AlgImprove}).
% \textcolor{red}{edgar: past tense}
We run \mwupetsc using 64 MPI proccesses, each with 1 thread
and list the results for \mtxname{rgg-18} in Table~\ref{tab:NoStrategyTable}. We choose
this graph since we have exact solutions for all five graph problems, and the run time with standard step size 
is not too large. 
% While we do not list other graphs due to space constraints,
The speedups for other graphs are within an order of magnitude of the ones listed here.

\begin{table}[h]
\caption{Convergence of MWU to find $(1+\epsilon)$-relative ($\epsilon=0.1$) solution on 
\mtxname{rgg-18} with standard step strategy (\textbf{Std}),
binary search (\textbf{Bin}), and Newton's method
(\textbf{Nwt}). For the latter two step strategies, we use the
previous step size as the initial step size.}
\footnotesize 
\centering
\begin{tabular}{|l||l|l|l|l|l|l|l|l|}
\hline
           & \multicolumn{3}{c|}{\textbf{\# MWU iters}} & \multicolumn{2}{c|}{\begin{tabular}[x]{@{}l@{}}\textbf{Avg \# step}\\\textbf{size iters}\end{tabular}}  & \multicolumn{3}{c|}{\textbf{Time (sec)}}   \\ \hline
           & \textbf{Std} & \textbf{Bin} & \textbf{Nwt} & \textbf{Bin} & \textbf{Nwt} & \textbf{Std} & \textbf{Bin} & \textbf{Nwt} \\ \hline \hline
\matching   & 25477         & 13           & 13            & 8.31   &   4.86  & 79.3          & 0.87        & 0.94         \\ \hline
\bimatching & 28210         & 15           & 13            & 8.00 &  5.07  & 261           & 1.08         & 1.05          \\ \hline
\domset     & 18837         & 96           & 166           & 5.78  &  2.58   & 41.3          & 1.30         & 1.58          \\ \hline
\vcover     & 30531         & 76           & 110           & 5.93  &   2.68  &   106           & 1.99         & 2.19          \\ \hline
\densest    & 20021         & 21           & 18            & 8.00   &   4.79   &    170           & 0.60         & 0.47          \\ \hline
\end{tabular}
\label{tab:NoStrategyTable}
\end{table}

The results verify that a step size search strategy significantly reduces the number of MWU iterations compared to 
% \chandra{change "no line search" to "the standard step size prescribed in theoretical algorithms"} 
the standard step size prescribed in theoretical algorithms. 
% with Newton's method reducing the iteration count by 2170x for \bimatching. This translates to faster  run times, e.g., 242x and 249x speedup from binary search and Newton's method, respectively.
Since an MWU iteration tends to be more expensive than a search step iteration (due to the SpMV), these results suggest that finding accurate step sizes, at the expense of a higher search cost, reduces the overall run time.
The performance difference between binary search or Newton's method is relatively small. While Newton's method on average requires fewer step size search iterations than binary search, it has more MWU iterations than binary search for the two pure covering problems, \domset and \vcover.
The additional MWU iterations observed when using Newton's method may be attributed to the $(1-\epsilon)$ multiplicative decrease (where $\epsilon=0.1$) applied to step sizes violating the bang-for-buck inequality~\eqref{eq:b4b_conv}.

\subsection{Effect of Software Optimizations}
% In Section~\ref{sec:sec:step_strat_exp}, we establish the performance improvement with line search. 
% that the \textcolor{red}{Edgar: line search? benefit from Newton's method over binary search seems negligible/unclear.} Newton Search mechanism improves the convergence significantly. 
We now analyze the acceleration of an MWU iteration with our software optimizations. To do so, we will compare the execution times of \mwupetsc and \mwuomp implementations with 68 threads. Later, we will also compare the execution times of \mwupetsc and \mwuhybrid implementations for problems with implicit matrix vector multiplication.
%\caleb{We should define \mwuhybrid here or earlier, like Section 6.2: Implementations}
% We parallelize MWU through two different means: PETSc and OpenMP.
% PETSc provides a suite of data structures and routines for the parallel solutions of scientific applications and employs MPI for its communication protocols~\cite{}. We make use of only its matrix and vector libraries and their respective operations. For our OpenMP implementation, we incorporate all the optimizations detailed
% in Section~\ref{sec:optimizations}.  We now test the improvement gain we get using these optimizations.

% \serif{add with 68 threads}
\subsubsection{Performance Breakdown}

First, we consider where the cycles are spent in our \mwupetsc implementation. Figure~\ref{fig:breakdown_petsc} shows the fraction of time spent in matrix-vector products (\emph{matvec}), step size search (\emph{search}), and other vector operations (\emph{vec}).
% \serif{will revisit names after section4 done} 
The gradients (Line~\ref{line:smax},~\ref{line:smin}) and new direction (Line~\ref{line:direction}) are included in the \emph{vec} category while all other vector operations done during step size search are included in the search category.
% \josep{fig 8a is missing the legend}

We observe that both applications and input graphs affect the distribution of execution cycles among these three components. For example, while \matching, \bimatching and \domset problems spend most of their time during \emph{matvec}, \vcover and \densest problems spend more than 50\% of their execution time for \emph{vec} and \emph{search} operations. \emph{matvec} takes on average 75\%, 78\%, and 82\% of the execution time for \matching, \bimatching, and \domset problems, respectively. In contrast, for \vcover and \densest, \emph{matvec} takes only 45\% and 38\% of the execution time on average.
% . We also observed that search time significantly increases with rgg-24 input for adj-matching and dom-set.\serif{will rewrite this}
% Note that search is composed of many vector operations. 
Due to this variable behavior, it is crucial to optimize both matrix-vector multiplications and vector operations for MWU.
\begin{figure}[h]
\vspace*{-2ex}
     \centering
    \captionsetup[subfloat]{farskip=0pt,captionskip=0pt}
    \subfloat[\mwupetsc execution time breakdown]{\includegraphics[width=0.45\textwidth]{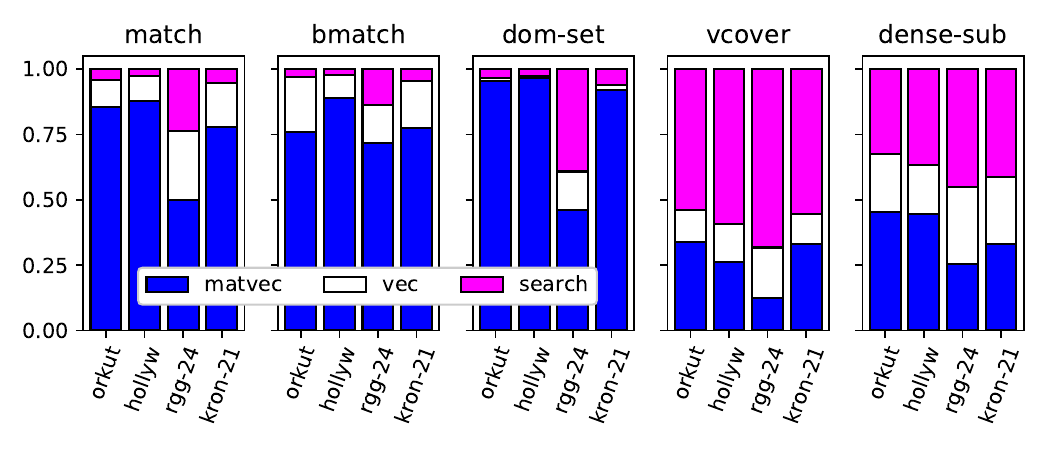}\label{fig:breakdown_petsc}}
   \\
    % \captionsetup[subfloat]{farskip=0pt,captionskip=0pt}
    \subfloat[OMP Breakdown]{\includegraphics[width=0.45\textwidth]{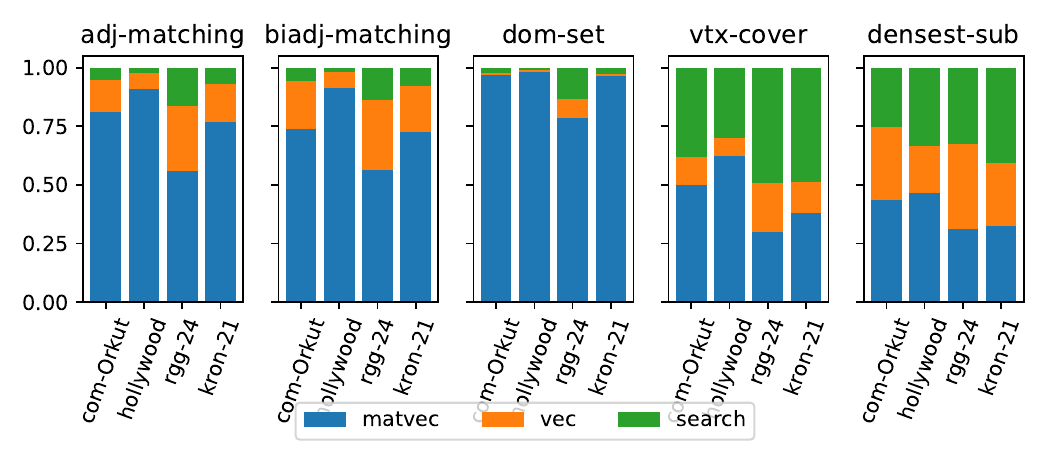}}
  \\

    \captionsetup[subfloat]{farskip=0pt,captionskip=0pt}
    \subfloat[Speedups of MWU components with \mwuomp compared to \mwupetsc]{\includegraphics[width=0.45\textwidth]{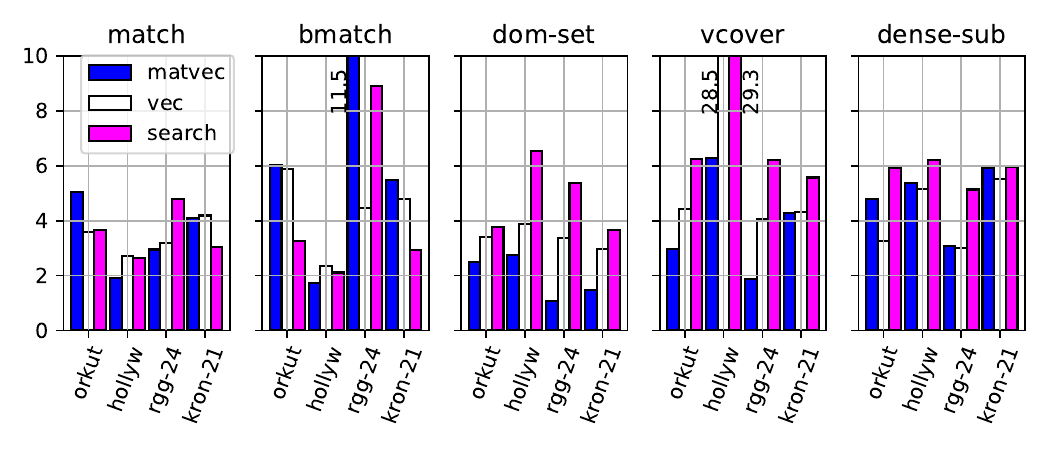}\label{fig:opt_speedup}}
   %\vspace{1mm}
    \caption{Breakdown of execution times and speedups obtained with \mwuomp for different components.}
    % \serif{will update the titles}
    \label{fig:breakdown}
    \vspace*{-3ex}
\end{figure}

% \serif{Depending on the problem... search+vector is expensive (when we use incidence matrices)}

\subsubsection{Shared-Memory Performance Optimizations}

Figure~\ref{fig:opt_speedup} shows the speedup obtained by our optimized
implementation relative to the PETSc-based implementation when executing on a
single node.
In this section, we report geometric mean speedups when referring to
average speedup across graphs. For \domset, the speedup is obtained from using a
favorable format (CSB) and minimizing memory accesses. Our optimizations
accelerate \emph{matvec} operations by 1.8x on average. Although \emph{vec} and
\emph{search} operations also get speedups, their contribution to overall
performance is smaller. On the other hand, for \matching, \bimatching, and
\vcover problems, we can observe the benefit of specialized vertex-incidence
matrix vector multiplications. In these cases, \emph{matvec} operations are
3.28x, 5.06x, and 3.49x faster on average, respectively.
% with specialized implementations. 
For  \densest problem, we also see benefits of vertex-edge pair matrix and interweaved-identity matrix specializations. \emph{matvec} operations are 4.64x faster on average. 

Moreover, \vcover and \densest problems spend a large amount of time for \emph{vec} and \emph{search} operations. We see that, in these cases, \mwuomp implementation can obtain significant speedups for both \emph{vec} (6.85x, and 4.08x  on average, respectively) and \emph{search} (8.92x and 5.79x on average, respectively) thanks to fusing and SIMD optimizations. 
% \serif{need to add the corner cases. these are because of slightly different M values found probably coming from fp precision difference}

% \serif{dominating set --> minimizing memory}, \serif{matching, vtx-cover --> incidence matrix emulation + vectorization and fusing}, \serif{densest subgraph --> oddshifted + interweaved + vectorization}

% \mina{For the distributed problem, we observe that for 4 nodes, the implicit matrix-vector product optimizations accelerates \emph{matvec} operations in \mwuhybrid by 1.4-3x compared to \mwupetsc for all graphs except \mtxname{rgg-24}, for which it slows down by 9x. This difference in speedup increases to 515x for 64 processes, though \mtxname{rgg-24} stagnates at around 11-24x slower in \mwuhybrid. All experiments are run with 64 OMP threads per MPI process. Table \ref{tab:impmatvec_dist_speedup} contains the ratio of \mwuhybrid run time to \mwupetsc run time for each configuration. In parenthesis is the ratio of \emph{matvec} product time to \emph{matvec} communication time for \mwuhybrid.}

\subsubsection{Distributed-Memory Optimizations}

%\caleb{Caleb: Re-ordered this paragraph. We should also reformat table to look like previous ones (this one has a lot of space)}
We record run time improvements in the context of multi-node exeuction \mwuhybrid over \mwupetsc in Table~\ref{tab:impmatvec_dist_speedup}. In parenthesis is the ratio of \emph{matvec} product time to \emph{matvec} communication time for \mwuhybrid. All experiments are run with 64 OMP threads per MPI process. We observe that for 4 nodes, the use of implicit matrix-vector products accelerates \emph{matvec} operations in \mwuhybrid by 1.4-3x compared to \mwupetsc for all graphs except \mtxname{rgg-24}, for which it is slower by 9x. \mina{As previously discussed in Section \ref{sec:MWUScale}, \mwupetsc uses a 1D communication layout, which is more efficient on banded matrices like \mtxname{rgg-24} than our 2D communication layout.}

%\mwupetsc uses a 1D communication layout, which for most graphs requires more communication that the implicit 2D communication layout used in \mwuhybrid. For the exception, \mtxname{rgg-24}, the structure of the adjacency matrix and the incidence matrix is a sparse matrix with small bandwidth, and as previously discussed in Section \ref{sec:MWUScale}, PETSc has an efficient internal representation for communication on these types of matrices.}
%The good performance and scalability achieved by PETSc 
%This difference in speedup increases to 515x for 64 processes, though \mtxname{rgg-24} stagnates at around 11-24x slower in \mwuhybrid. 
% \vspace*{-1ex}

\begin{table}[h]
\footnotesize
\caption{Speed-up in run time of our implicit implementation of the product of the edge-incidence matrix and a vector. The ratio of computation to communication time in \mwuhybrid is parenthesized.}
\centering
 \begin{tabular}{||c | c c c c ||} 
 \hline
 & \mtxname{hollyw} & \mtxname{orkut} & \mtxname{rgg-24} & \mtxname{kron-21} \\ [0.5ex] 
 \hline\hline
 %\mwupetsc 4n & 22.4 & 25.5 & 3.87 & 29.7 \\
 %\mwupetsc 4n & 22 & 26 & 3.9 & 30 \\ 
 %\hline
 %\mwuhybrid 4n & 5.7+10. & 4.5+4.0 & 2.9+32 & 4.5+5.4\\
 4 nodes & 1.4 (0.57) & 3 (1.1) &  0.11 (0.09) & 3 (0.83) \\
 %\mwupetsc 16n & 590 & 370 & 2.4 & 950 \\ 
 \hline
 %\mwuhybrid 16n & 1.4+3.6 & 1.3+6.7 & 0.89+55 & 1.5+8.3\\
 16 nodes & 120 (0.39) & 46 (0.19) &  0.04 (0.02) & 97 (0.18) \\
  \hline
 %\mwupetsc 64n & 410 & 840 & 2.2 & 1470 \\ 
 %\hline
 %\mwuhybrid 64n & 0.51+1.7 & 0.34+2.9 & 0.29+23 & 0.45+2.4\\
 64 nodes & 186 (0.3) & 259 (0.12) &  0.09 (0.01) & 516 (0.19) \\
 \hline
\end{tabular}
\label{tab:impmatvec_dist_speedup}
\end{table}

\section{Conclusion}
%\sout{We implement a shared and distributed-memory implementation of the algorithm from Mahoney et. al.~\cite{mahoney2016approximating}, which we call MWU.  We accelerate the method by incorporating a step size search and optimizing its sparse linear algebra  operations. When applied to various graph problems, MWU finds $(1+ \epsilon)$-relative ($\epsilon=0.1$) solutions much faster general than general LP solvers CPLEX and Gurobi, both of which obtain exact solutions.   MWU can also run faster than a specialized parallel algorithm for bipartite matching~\cite{azad2016computing} (which computes exact solutions) on graphs with planar-structures, but it does relatively poorly on graphs with community-structures.} % Overall, our work demonstrates the practicality of approximately solving optimization problems via parallel solvers for positive LPs.

%\sout{In addition, MWU can easily be extended to solve optimization problems beyond the ones in this paper. For example,  solving the unweighted matching problem is as easy to solve as the weighted case using the  LP framework. In contrast, the custom-designed code of~\cite{azad2016computing}  achieves fast parallel performance for unweighted bipartite matching, but it is not  clear how or whether their algorithm and performance can generalize to the weighted case. }

Our work 
%provides evidence
demonstrates that approximate positive LP solvers are an efficient and scalable approach for solving a wide range of graph problems.
We show that with carefully chosen modifications and implementation of the MWU algorithm from Mahoney et. al.~\cite{mahoney2016approximating} -- namely, a step size search strategy and specialized linear algebra operations that leverage shared and distributed-memory resources -- the algorithm exceeds the performance of general purpose LP solvers for finding a $(1+\epsilon)$-relative solution. 
%Surprisingly, 
Our implementation also matches the performance of hand-tuned parallel graph libraries for some graphs. 
%\sout{Naturally, one can ask if these general purpose, approximate positive LP solvers can match the latter's performance. To do so, one must incorporate further speed ups. For instance, when solving maximum matching, one can warm-start MWU by setting $x_0$ to be a greedy solution, as done in~\cite{azad2016computing}. However, an immediate hurdle is that MWU only monotonically increases elements of $\B x$, whereas in gradient descent elements can both increase and decrease.}
%Exploring other step search strategies and update vectors, especially those that return high-quality steps at the expense of a more expensive search, may yield further performance improvements.
%Finally, we note that many positive LP problems contain as constraint matrices either an adjacency matrix, a vertex-edge incidence matrix, or some concatenation of these matrices. 
%Furthermore, generalizing our mechanisms for implicit SpMVs to handle block-structured matrices would permit efficient handling of positive LPs, such as stochastic vertex and edge cover~\cite{ravi2006hedging}.

\begin{acks}
    This research has been supported by funding from the United States National
    Science Foundation (NSF) via grant \#1942995 and \#2016136, as well as NSF
    grant CCF-1910149. This material is based upon work supported by the U.S.
    Department of Energy, Office of Science, Office of Advanced Scientific
    Computing Research, Department of Energy Computational Science Graduate
    Fellowship under Award Number DE-SC0022158. 
\end{acks}


\bibliographystyle{ACM-Reference-Format}
\bibliography{sample-base}

\clearpage
\appendix

\section{Supplementary Material}
\subsection{Further Details on Generalized Matching Experiments}
\label{sec:more_details_on_gm}
Let $G=(V,E)$ be an undirected, unweighted graph. For generalized matching, a vertex $v$ can be matched $b(v)$ times, where $\mathrm{lb}(v) \leq b(v) \leq \mathrm{ub}(v)$ are lower and upper bounds on the number of unique vertices matching with $v$. More precisely, the IP formulation is
\begin{equation} \label{eq:gm_ip}
\begin{split}
    \exists \ \B x \text{ s.t. } & \ \mathrm{lb}(v) \leq \sum\limits_{e \in \mathrm{inc}(v)} x_e \leq \mathrm{ub}(v), \forall v \in V, \\
    & \ x_e \in \{0,1\}, \forall e \in E.
\end{split}
\end{equation}
Maximum matching is equivalent to generalized matching with $\mathrm{lb}(v)=0$ and $\mathrm{ub}(v)=1, \forall v \in V$, as well as a (maximization) objective function of $\sum\limits_e x_e$.
The LP relaxation is the feasibility mixed packing and covering LP,
\begin{equation*}
    \exists \ \B x \in \mathbf{R}^m \text{ s.t. } \B{Mx} \geq \B{l}, \B{Mx} \leq \B{u}, \B x \geq \mathbb{0}, 
\end{equation*}
where $\B M$ is the vertex-edge incidence matrix, and $\B l, \B u \in \mathbf{R}^n$ are the vector of lower and upper bounds for each vertex.

\subsection{Dataset Preprocessing}

Now, let us describe how to pre-process the Netflix~\cite{bennett2007netflix} and KDD~\cite{dror2012yahoo} datasets, as detailed in~\cite{makari2013distributed}. Both datasets contain users and items (e.g., movies in Netflix, music tracks in KDD) as vertices, and edges correspond to a user rating an item. This dataset is represented as a bipartite graph, where users and items form the two partitions, and edges go only between vertices in separate partitions. For the number of matchings with each user, we enforce a lower bound of three and upper bound of five. For items, no lower bound is set, but an upper bound of 200 and 2000 is chosen for Netflix and KDD, respectively. Finally, to ensure there is a feasible matching satisfying these bounds, we exclude users with less than ten ratings from the Netflix dataset. After this pre-processing step, the two datasets have 473k and 1.6m vertices, as well as 100m and 252m edges, respectively. 

\subsection{Gradient Descent with Adaptive Error}

Finally, we review the gradient descent algorithm with an adaptive error of~\cite{makari2013distributed}. In short, the algorithm minimizes the convex function via gradient descent,
\[
    \Gamma(\B x) = \sum\limits_{i=1}^{m_P} y_i(\B x) + \sum\limits_{i=1}^{m_C} z_i(\B x),
\]
where for some $\mu > 0$, 
\begin{align*}
    y_i(\B x) &= \mathrm{exp}\left [\mu \cdot (\B{P}_i\B{x}-1) \right ] \\
    z_i(\B x) &= \mathrm{exp}\left [\mu \cdot (1-\B{C}_i\B{x}) \right ] .
\end{align*}
The algorithm contains two error values. There is the error bound $\epsilon$, where the algorithm seeks to find an $\B x$ that is a $(1+\epsilon)$-relative solution. Then there is the \emph{internal} error bound $\epsilon'$, which is used to specify $\mu$ as well as which coordinates of $x_i$ to update, and by how much. The authors of~\cite{makari2013distributed} found that they can set $\epsilon' > \epsilon$. For example, when $\epsilon=0.05$, they can choose $\epsilon'=1$. Then they run the algorithm until it stagnates, and if $\B x$ is not an $(1+\epsilon)$-relative solution, they decrement $\epsilon'$ and warm-start the algorithm by setting the initial $\B{x}_0$ to the solution of the previous, stagnated algorithm. This strategy is called \emph{adaptive error}, since it adaptively updates the internal error bound. 

\end{document}